\documentclass[11pt]{article}
 
\usepackage[margin=1in]{geometry}
\usepackage{amsmath,amsthm,amssymb}
\usepackage{mathtools}
\usepackage{hyperref} 
\hypersetup{colorlinks,linkcolor={blue},citecolor={blue},urlcolor={red}}
\usepackage{algorithm}
\usepackage{algpseudocode}

\newcommand{\f}{\mathbb{F}_2}

\newcommand{\wh}{\widehat}
\newcommand{\B}{\mathcal{B}}
\newcommand{\R}{\mathcal{R}}
\newcommand{\E}{\mathbb{E}}
\newcommand{\Z}{\mathbb{Z}}
\newcommand{\G}{\mathcal{G}}

\newcommand{\C}{\mathcal{C}}
\renewcommand{\d}{\mathsf{d}}
\renewcommand{\S}{\mathcal{S}}
\newcommand{\1}{\mathbf{1}}
\renewcommand{\P}{\mathcal{P}}

\renewcommand{\int}{\mathrm{int}}
\newcommand{\ADD}{\mathrm{ADD}}
\newcommand{\PDT}{\mathrm{PDT}}

\newcommand{\zone}{\{0, 1\}}
\newcommand{\pmone}{\{-1, 1\}}
\newcommand{\reals}{\mathbb{R}}
\newcommand{\bra}[1]{\left\{#1\right\}}

\newcommand{\wdeg}{\widetilde{\mathsf{deg}}}

\newcommand{\mathify}[1]{\ifmmode{#1}\else\mbox{$#1$}\fi}
\newcommand{\abs}[1]{\mathify{\left| #1 \right|}}

\newtheorem{theorem}{Theorem}[section]
\newtheorem{corollary}[theorem]{Corollary}
\newtheorem{definition}[theorem]{Definition}
\newtheorem{proposition}[theorem]{Proposition}
\newtheorem{observation}[theorem]{Observation}

\newtheorem{lemma}[theorem]{Lemma}

\newtheorem{claim}[theorem]{Claim}
\newtheorem{question}[theorem]{Question}
\newtheorem{conjecture}[theorem]{Conjecture}

\allowdisplaybreaks

\title{On parity decision trees for Fourier-sparse Boolean functions}

\author{
Nikhil S.~Mande\\
Georgetown University\\
\textsf{nikhil.mande@georgetown.edu}
\and
Swagato Sanyal\footnote{S.S.~is supported by an ISIRD Grant from Sponsored Research and Industrial Consultancy, IIT Kharagpur.}\\
IIT Kharagpur\\
\textsf{swagato@cse.iitkgp.ac.in}}
\date{}

\begin{document}

\maketitle

\begin{abstract}
We study parity decision trees for Boolean functions. The motivation of our study is the log-rank conjecture for XOR functions and its connection to Fourier analysis and parity decision tree complexity. Our contributions are as follows. Let $f : \f^n \to \pmone$ be a Boolean function with Fourier support $\S$ and Fourier sparsity $k$.
\begin{itemize}
    \item We prove via the probabilistic method that there exists a parity decision tree of depth $O(\sqrt{k})$ that computes $f$. This matches the best known upper bound on the parity decision tree complexity of Boolean functions (Tsang, Wong, Xie, and Zhang, FOCS 2013). Moreover, while previous constructions (Tsang et al., FOCS 2013, Shpilka, Tal, and Volk, Comput.~Complex.~2017) build the trees by carefully choosing the parities to be queried in each step, our proof shows that a naive sampling of the parities suffices.
    \item We generalize the above result by showing that if the Fourier spectra of Boolean functions satisfy a natural ``folding property'', then the above proof can be adapted to establish existence of a tree of complexity polynomially smaller than $O(\sqrt k)$. More concretely, the folding property we consider is that for most distinct $\gamma, \delta$ in $\S$, there are at least a polynomial (in $k$) number of pairs $(\alpha, \beta)$ of parities in $\S$ such that $\alpha+\beta=\gamma+\delta$.
    We make a conjecture in this regard which, if true, implies that the communication complexity of an XOR function is bounded above by the fourth root of the rank of its communication matrix, improving upon the previously known upper bound of square root of rank (Tsang et al., FOCS 2013, Lovett, J.~ACM.~2016).
    \item Motivated by the above, we present some structural results about the Fourier spectra of Boolean functions.
    It can be shown by elementary techniques that for any Boolean function $f$ and all $(\alpha, \beta)$ in $\binom{\S}{2}$, there exists another pair $(\gamma, \delta)$ in $\binom{\S}{2}$ such that $\alpha + \beta = \gamma + \delta$. One can view this as a ``trivial'' folding property that all Boolean functions satisfy. Prior to our work, it was conceivable that for all $(\alpha, \beta) \in \binom{\S}{2}$, there exists exactly one other pair $(\gamma, \delta) \in \binom{\S}{2}$ with $\alpha + \beta = \gamma + \delta$. We show, among other results, that there must exist several $\gamma \in \f^n$ such that there are at least three pairs of parities $(\alpha_1, \alpha_2) \in \binom{\S}{2}$ with $\alpha_1+\alpha_2=\gamma$. This, in particular, rules out the possibility stated earlier.
\end{itemize} 
\end{abstract}

\section{Introduction}
\label{sec:introduction}

The log-rank conjecture~\cite{LS88} is a fundamental unsolved question in communication complexity that states that the deterministic communication complexity of a Boolean function is polynomially related to the logarithm of the rank (over real numbers) of its communication matrix. The importance of the conjecture stems from the fact that it proposes to characterize communication complexity, which is an interactive complexity measure, by the rank of a matrix which is a traditional and well-understood algebraic measure. In this work we focus on the important and well-studied class of XOR functions. Consider a two-party function $F:\f^n\times\f^n \rightarrow\pmone$ whose value on any input $(x,y)$ depends only on the bitwise XOR of $x$ and $y$, i.e., there exists a function $f: \f^n \rightarrow\pmone$ such that for each $(x,y) \in \f^n$, $F(x,y)=f(x \oplus y)$. Such a function $F$ is called an XOR function, and is denoted as $F = f \circ \oplus$. The log-rank conjecture and communication complexity of such an XOR function $F$ has interesting connections with the Fourier spectrum of $f$. For example, it is known that the rank of the communication matrix of $F$ equals the Fourier sparsity of $f$ (henceforth referred to as $k$)~\cite{BC99}. The natural randomized analogue of the log-rank conjecture is the log-approximate-rank conjecture~\cite{LS09}, which was recently refuted by Chattopadhyay, Mande, and Sherif~\cite{CMS19}. The quantum analogue of the log-rank conjecture was subsequently also refuted by Sinha and de Wolf~\cite{SdW19} and Anshu, Boddu, and Touchette~\cite{ABT19}. It is worth noting that an XOR function was used to refute these conjectures.
 
 To design a cheap communication protocol for $F$, an approach adopted by many works~\cite{STV17, TWXZ13, MO09} is to design a small-depth parity decision tree (henceforth referred to as PDT) for $f$, and having a communication protocol simulate the tree; it is easy to see that the parity of a subset of bits of the string $x \oplus y$ can be computed by the communicating parties by interchanging two bits. The parity decision tree complexity (henceforth referred to as PDT($\cdot$)) of $f$ thus places an asymptotic upper bound on the communication complexity of $F$. The work of Hatami, Hosseini and Lovett~\cite{HHL18} shows that this approach is polynomially tight; they showed that $\PDT(f)$ is polynomially related to the deterministic communication complexity of $F$. In light of this, the log-rank conjecture for XOR functions $F=f \circ \oplus$ is readily seen to be equivalent to $\PDT(f)$ being polylogarithmic in $k$.
 
However, we are currently very far from achieving this goal. Lovett~\cite{Lov16} showed that the deterministic communication complexity of any Boolean function $F$ is bounded above by $O(\sqrt \mathsf{rank}(F)\log \mathsf{rank}(F))$. In particular, this implies that that the deterministic communication complexity of $F = f \circ \oplus$ is $O(\sqrt k \log k)$. Improving upon a work of Shpilka et al.~\cite{STV17}, Tsang et al.~\cite{TWXZ13} showed that $\PDT(f) = O(\sqrt{k})$. In addition to bounding $\PDT(f)$ instead of the communication complexity of $F$, Tsang et al.~achieved a quantitative improvement by a logarithmic factor over Lovett's bound for the class of XOR functions. Sanyal~\cite{San19} showed that the simultaneous communication complexity of $F$ (characterized by the Fourier dimension of $f$) is bounded above by $O(\sqrt k \log k)$, and is tight (up to the $\log k$ factor) for the addressing function. 

In this work we derive new understanding about the structure of Fourier spectra of Boolean functions. Aided by this insight we reprove the $O(\sqrt{k})$ upper bound on $\PDT(f)$ (see Sections~\ref{section:tsang} and~\ref{section:result1}). We conditionally improve this bound by a polynomial factor, assuming a ``folding property'' of the Fourier  spectra of Boolean functions (see Section~\ref{section:result2}). To prove these results, we make use of a simple necessary condition for a function to be Boolean (see Proposition~\ref{prop:tit}). While we show that it is not a sufficient condition (see Theorem~\ref{thm: ceg} in Appendix~\ref{sec: app}), it does enable us to prove the above results. In these proofs, we use Proposition~\ref{prop:tit} in conjunction with probabilistic and combinatorial arguments. Finally, we make progress towards establishing the folding property (see Section~\ref{section:result3}). Here we use the well-known characterization of Boolean functions given by two conditions, namely Parseval's identity (Equation~(\ref{eqn: parseval})) and a condition attributed to Titsworth (Equation~(\ref{eqn: boolstruct})), in conjunction with combinatorial arguments.

\subsection{Organization of this paper}
In Section~\ref{sec: prelims} we review some preliminaries and introduce the notation that we use in this paper. In this section we also introduce definitions and concepts that are needed to state our results formally. In Section~\ref{sec: contribs} we motivate and formally state our results, and discuss proof techniques. The formal proofs of our main results can be found in Sections~\ref{sec:firstthm},~\ref{sec: improved}, and~\ref{sec: nontrivial}.

\section{Notation and preliminaries}\label{sec: prelims}

All logarithms in this paper are taken with base 2. As is standard, we use the notation $f(n)=\widetilde{O}(h(n))$ ($f(n)=\widetilde{\Theta}(\cdot), f(n)=\widetilde{\Omega}(\cdot)$) to convey that there exists a constant $c\geq0$ such that that $f(n)=O(h(n)\log^c h(n))$ ($f(n)=\Theta(h(n)\log^c h(n))$, $f(n)=\Omega (h(n)\log^c h(n))$, respectively). We use the notation $[n]$ to denote the set $\bra{1, 2, \dots, n}$.  For any set $S$, we use the notation $\binom{S}{2}$ to denote the set of all subsets of $S$ of size exactly $2$. We abuse notation and denote a generic element of $\binom{S}{2}$ as $(a, b)$ rather than $\bra{a, b}$. When we use the notation $\E_{x \in X}[\cdot]$, the underlying distribution corresponds to $x$ being sampled uniformly at random from $X$. We use the symbol ``$+$'' to denote both coordinate-wise addition over $\f$ as well as addition over reals; the meaning in use will be clear from context. For sets $A, B \subseteq \f^n$, $A+B$ denotes the \emph{sumset} defined by $\bra{\alpha + \beta ~\middle|~  \alpha \in A, \beta \in B}$. For a set $A \subseteq \f^n$ and $\gamma \in \f^n$, we denote by $A + \gamma$ the set $A + \bra{\gamma}$. The above convention also extends to the symbol ``$\sum$''. For a set of vectors $\Gamma \in \f^n$, we define $\mathsf{span}~\Gamma$ to be the set of all $\f$-linear combinations of vectors in $\Gamma$, i.e., $\mathsf{span}~\Gamma=\bra{\sum_{\gamma \in \Gamma} c_\gamma\cdot \gamma ~\middle|~ c_\gamma \in \f\mbox{\ for\ }\gamma \in \Gamma}$.

Consider the vector space of functions from $\f^n$ to $\reals$, equipped with the following inner product.
\[
\langle f, g \rangle := \E_{x \in \f^n}[f(x)g(x)] = \frac{1}{2^n}\sum_{x \in \f^n}f(x)g(x).
\]
Let $x=(x_1,\ldots,x_n) \in \f^n$. For each $\alpha=(\alpha_1,\ldots,\alpha_n) \in \f^n$, define $\alpha(x) := \sum_{i = 1}^n\alpha_ix_i$ (mod 2), and the associated \emph{character} $\chi_\alpha : \f^n \to \pmone$ by $\chi_\alpha(x) := (-1)^{\alpha(x)}$. Observe that $\chi_\alpha(x)$ is the $\pm 1$-valued parity of the bits $\bra{x_i ~\middle|~ \alpha_i=1}$; due to this we will also refer to characters as parities. The set of parities $\bra{\chi_\alpha ~\middle|~  \alpha \in \f^n}$ forms an orthonormal (with respect to the above inner product) basis for this vector space. Hence, every function $f : \f^n \to \reals$ can be uniquely written as $f = \sum_{\alpha \in \f^n}\wh{f}(\alpha) \chi_\alpha$, where $\wh{f}(\alpha) = \langle f, \chi_\alpha \rangle = \E_{x \in \f^n} [f(x)\chi_\alpha(x)]$.
The coefficients $\bra{\wh{f}(\alpha) ~\middle|~ \alpha \in \f^n}$ are called the Fourier coefficients of $f$.

For any function $f : \f^n \to \pmone$ and any set $A \subseteq \f^n$, define the function $f \mid_{A} : A \to \pmone$ by $f \mid_{A}(x) = f(x)$ for all $x \in A$. In other words, $f \mid_A$ denotes the restriction of $f$ to $A$.

Throughout this paper, for any Boolean function $f : \f^n \to \pmone$, we denote by $\S$ the Fourier support of $f$, i.e.~$\S = \bra{\alpha \in \f^n ~\middle|~ \wh{f}(\alpha) \neq 0}$. We also denote by $k$ the Fourier sparsity of $f$, i.e.~$k = |\S|$. The dependence of $\S$ and $k$ on $f$ is suppressed and the underlying function will be clear from context.

The representation of Fourier coefficients as an expectation (over $x \in \f^n$) immediately yields the following observation about granularity of Fourier coefficients of Boolean functions.

\begin{observation}\label{obs: gran}
Let $f : \f^n \to \pmone$ be any Boolean function. Then, for all $\alpha \in \f^n$, $\wh{f}(\alpha)$ is an integral multiple of $1/2^n$.
\end{observation}

We next define plateaued functions.
\begin{definition}[Plateaued functions]
A Boolean function $f : \f^n \to \pmone$ is said to be \emph{plateaued} if there exists $x \in \reals$ such that $\wh{f}(\alpha) \in \bra{0, x, -x}$ for all $\alpha \in \f^n$.
\end{definition}
Next we define the addressing function.
\begin{definition}[Addressing function]
\label{def:addressing}
Let $k$ be an even power of $2$. The addressing function $\mathsf{ADD}_k:\f^{\frac{1}{2}\log k + \sqrt{k}} \to \pmone$ is defined as
\[
\mathsf{ADD}_k(x,y_1,\ldots,y_{\sqrt{k}}):=(-1)^{y_{\mathrm{int}(x)}},
\]
where $x \in \f^{\frac{1}{2}\log k}, y_i \in \f$ for $i=1,\ldots,\sqrt{k}$, and $\mathrm{int}(x)$ is the unique integer in $\bra{1, \dots, \sqrt{k}}$ whose binary representation is $x$.
\end{definition}
The Fourier sparsity of $\mathsf{ADD}_k$ can be verified to be $k$.
We now define a notion of equivalence on elements of $\binom{\S}{2}$.
\begin{definition}
For any Boolean function $f : \f^n \to \pmone$, we say a pair $(\alpha_1, \alpha_2) \in \binom{\S}{2}$ is \emph{equivalent} to $(\alpha_3, \alpha_4) \in \binom{\S}{2}$ if $\alpha_1 + \alpha_2 = \alpha_3 + \alpha_4$.
\end{definition}
In the above definition, if $\alpha_1 + \alpha_2 = \alpha_3 + \alpha_4 = \gamma$, then we say that the pairs $(\alpha_1, \alpha_2)$ and $(\alpha_3, \alpha_4)$ \emph{fold in the direction} $\gamma$.
We also say that the elements $\alpha_1, \alpha_2, \alpha_3$, and $\alpha_4$ \emph{participate in the folding direction $\gamma$}.
It is not hard to verify that the notion of equivalence defined above does indeed form an equivalence relation. We will denote by $O_\gamma$ the equivalence class of pairs that fold in the direction $\gamma$, i.e.,
\[
O_\gamma := \bra{(\alpha, \beta) \in \binom{\S}{2} ~\middle|~  \alpha + \beta = \gamma}.
\]
We suppress the dependence of $O_\gamma$ on the underlying function $f$, which will be clear from context. 
Unless mentioned otherwise, these are the equivalence classes under consideration throughout this paper.

For any Boolean function $f : \f^n \to \pmone$, we have for each $x \in \f^n$:
\begin{equation}
1 = f^2(x) = \sum_{\gamma \in \f^n} \left(\sum_{(\alpha_1, \alpha_2) \in \f^n \times \f^n: \alpha_1+\alpha_2=\gamma}\widehat{f}(\alpha_1)\widehat{f}(\alpha_2)\right)\chi_\gamma(x). \label{eqn:coefmch}
\end{equation}

Matching the constant term of each side of the above identity we have 
\begin{equation}\label{eqn: parseval}
\sum_{\alpha \in \f^n}\wh{f}(\alpha)^2 = 1,
\end{equation}
which is commonly referred to as Parseval's identity for Boolean functions. By matching the coefficient of each non-constant $\chi_\gamma$ on each side of Equation~\eqref{eqn:coefmch} we obtain
\begin{equation}\label{eqn: boolstruct}
\forall \gamma \neq \emptyset, \sum_{(\alpha_1, \alpha_2) \in \f^n \times \f^n : \alpha_1+\alpha_2=\gamma}\widehat{f}(\alpha_1)\widehat{f}(\alpha_2)=0.
\end{equation}
Equation~\eqref{eqn: boolstruct} is attributed to Titsworth~\cite{Tit62}.
The following proposition is an easy consequence of Equation~\eqref{eqn: boolstruct}. It provides a necessary condition for a subset of $\f^n$ to be the Fourier support of a Boolean function.
\begin{proposition}\label{prop:tit}
Let $f : \f^n \to \pmone$ be a Boolean function.
Then, for all $(\alpha, \beta) \in \binom{\S}{2}$, there exists $(\gamma, \delta) \neq (\alpha, \beta) \in \binom{\S}{2}$ such that $\alpha + \beta = \gamma + \delta$.
In other words, $|O_{\alpha+\beta}| \geq 2$.
\end{proposition}

The Fourier $\ell_1$-norm of $f$ is defined as $\|\wh{f}\|_1 := \sum_{\alpha \in \f^n} |\wh{f}(\alpha)|$.
By the Cauchy-Schwarz inequality and Equation~\eqref{eqn: parseval}, we have
\begin{equation}\label{eqn: 1normk}
\|\wh{f}\|_1 \leq \sqrt{k}\sqrt{\sum_{\alpha \in \f^n}\wh{f}(\alpha)^2} = \sqrt{k}.
\end{equation}

We next formally define parity decision trees.

A \emph{parity decision tree} (PDT) is a binary tree whose leaf nodes are labeled in $\pmone$, each internal node is labeled by a parity $\chi_\alpha$ and has two outgoing edges, labeled $-1$ and $1$.
On an input $x \in \f^n$, the tree's computation proceeds from the root down as follows: compute $\chi_\alpha(x)$ as indicated by the node's label and following the edge indicated by the value output, and continue in a similar fashion until a reaching a leaf, at which point the value of the leaf is output. When the computation reaches a particular internal node, the PDT is said to \emph{query} the parity label of that node.
The PDT is said to compute a function $f : \f^n \to \pmone$ if its output equals the value of $f$ for all $x \in \f^n$.
The parity decision tree complexity of $f$, denoted $\PDT(f)$ is defined as
\[
\PDT(f) := \min_{T : T~\text{is a PDT computing}~f} \textnormal{depth}(T).
\]

\subsection{Restriction to an affine subspace}\label{sec:basic}

In this section we discuss the effect of restricting a function $f : \f^n \to \reals$ to an affine subspace, on the Fourier spectrum of $f$.

\begin{definition}
A set $V \subseteq \f^n$ is called an \emph{affine subspace} if there exist linearly independent vectors $\ell_1, \dots, \ell_t \in \f^n$ and $a_1, \dots, a_t \in \f$ such that $V = \bra{x \in \f^n ~\middle|~ \ell_i(x) = a_i~\forall i \in \{1,\ldots,t\}}$.
$t$ is called the \emph{co-dimension} of $V$.
\end{definition}

Consider a set $\Gamma:=\bra{\gamma_1, \ldots, \gamma_t}$ of vectors in $\f^n$. Define the set $\mathcal{G}:=\mathsf{span}\ \Gamma$, and let $\mathcal{C}:=\bra{\G+\beta ~\middle|~ \beta \in \f^n,~(\G+\beta)\cap\mathcal{S}\neq\emptyset}$ to be the cosets of $\G$ that have non-trivial intersection with $\mathcal{S}$. For each $C \in \C$, let $\alpha(C)$ denote an arbitrary but fixed element in $C \cap \S$. In light of this, we write the Fourier transform of $f$ as
\begin{equation}
\label{eqn-grouping}
f(x)=\sum_{C \in \mathcal{C}}\left(\sum_{\gamma \in \G} \widehat{f}(\alpha(C) + \gamma) \chi_\gamma(x)\right)\chi_{\alpha(C)}(x),
\end{equation}
For any such fixed $C$, the value of the sum $\sum_{\gamma \in \G} \widehat{f}(\alpha(C) + \gamma) \chi_\gamma(x)$ that appears in Equation~\eqref{eqn-grouping} is determined by the values $\gamma_1(x), \ldots, \gamma_t(x)$. We denote this sum by $P_C(\gamma_1(x), \ldots, \gamma_t(x))$.

\noindent For $\mathbf{b}:=(b_1, \ldots, b_t) \in \f^t$, let $H_\mathbf{b}$ be the affine subspace $\bra{x \in \f^n ~\middle|~ \gamma_1(x)=b_1, \ldots, \gamma_t(x)=b_t}$. It follows immediately that the Fourier transform of $f\mid_{H_\mathbf{b}}$ is given by
\begin{equation}
\label{eqn-restr}
f\mid_{H_\mathbf{b}}(x)=\sum_{C \in \mathcal{C}} P_C(b_1, \ldots, b_t) \chi_{\alpha(C)}(x).
\end{equation}
In particular, for each $\mathbf{b}$, the Fourier sparsity of $f\mid_{H_\mathbf{b}}$ is bounded above by $|\mathcal{C}|$. 

We note here that each element in $\S$ is mapped to a unique element in $\mathcal C$. The elements of $\mathcal C$ can thus be thought of as buckets that form a partition of $\S$. Keeping this view in mind we define the following.

\begin{definition}[Bucket complexity]\label{defn: balti}
Let $f : \f^n \to \pmone$ be any Boolean function. Consider a set of vectors $\Gamma = \bra{\gamma_1, \ldots, \gamma_t}$ in $\f^n$. Let $\mathcal{G}:=\mathsf{span}\ \Gamma$, and let $\C$ denote the set of cosets of $\G$ that have non-empty intersection with $\S$, that is, $\mathcal{C}:=\bra{\G+\beta~\middle|~ \beta \in \f^n,~(\G+\beta)\cap\mathcal{S}\neq\emptyset}$. Define the \emph{bucket complexity} of $f$ with respect to $\G$, denoted $\B(f, \G)$, as
\[
\B(f, \G) = |\C|.
\]
\end{definition}
We now make the following useful observation, which follows from Equation~\eqref{eqn-restr}.
\begin{observation}
\label{obs:bucktet<sp}
Let $\Gamma$ and $\G$ be as in Definition~\ref{defn: balti}. Let $\mathbf{b}=(b_1,\ldots, b_t) \in \f^t$ be arbitrary. Let $V$ be the affine subspace $\bra{x \in \f^n ~\middle|~ \gamma_1(x)=b_1,\ldots,\gamma_t(x)=b_t}$. Let $k'$ be the Fourier sparsity of $f\mid_V$. Then $k' \leq \B(f,\G)$.
\end{observation}

\begin{definition}[Identification of characters]\label{defn: ident}
For $f, \G$, and $\C$ as in Definition~\ref{defn: balti} and any $\beta, \delta \in \S$, we say that \emph{$\beta$ and $\delta$ are identified with respect to $\G$} if $\beta+\delta \in \G$, or equivalently, if $\beta$ and $\delta$ belong to the same coset in $\C$.
\end{definition}

The following observation plays a key role in the results discussed in this paper.
\begin{observation}
\label{obs:identify}
Let $f, \G$ and $\C$ be as in Definition~\ref{defn: balti}. If there exists a set $L \subseteq \S$ of size $h$ such that each $\beta \in L$ is identified with some other $\delta \in \S$ with respect to $\G$, then $\mathcal B(f, \G) \leq k-\frac{h}{2}$.
\end{observation}
\begin{proof}
Since $|\overline{L}| = k-h$, there are at most $k-h$ cosets in $\C$ that contain at least one element from $\overline{L}$. Next, each coset in $\C$ that contains only elements from $L$ has at least 2 elements (by the hypothesis). Hence, the number of cosets containing only elements from $L$ is at most $h/2$. Combining the above two, we have that $|\C|\leq (k-h)+\frac{h}{2}=k-\frac{h}{2}$.
\end{proof}

\subsection{Folding properties of Boolean functions}\label{sec: foldingprelims}

\begin{definition}\label{defn: folding}
Let $f : \f^n \to \pmone$ be any Boolean function.
We say that $f$ is \emph{$(\delta, \ell)$-folding} if 
\[
\left|\bra{(\alpha, \beta) \in \binom{\S}{2} ~\middle|~ |O_{\alpha+\beta}| \geq k^{\ell} + 1}\right| \geq \delta\binom{k}{2}.
\]
\end{definition}
Proposition~\ref{prop:tit} implies that any Boolean function is $(1,0)$-folding.

We next show by a simple averaging argument that if $f$ has ``good folding properties'', then there are many $\alpha \in \S$, such that $|O_{\alpha+\beta}|$ is large for many $\beta \in \S \setminus \bra{\alpha}$.
\begin{claim}
\label{claim:folding}
Let $f : \f^n \to \pmone$ be $(\delta,\ell)$-folding with $k$ sufficiently large. Define
\[
U:=\bra{\alpha \in \S ~\middle|~ \mbox{there exist at least $\delta k/2$ many $\beta \in \S \setminus \bra{\alpha}$ with $|O_{\alpha+\beta}|\geq k^\ell+1$}}.
\]
Then $|U|\geq \frac{\delta k}{3}$.
\end{claim}

\begin{proof}
For each $\alpha \in \S$, define $t(\alpha):=|\bra{\beta \in \S \setminus \bra{\alpha} ~\middle|~ |O_{\alpha+\beta}|\geq k^\ell+1}|$. By the hypothesis, $\sum_{\alpha \in \S}t(\alpha) \geq \delta k(k-1)$. We have
\begin{align*}
& |U|\cdot k+(k-|U|)\cdot\frac{\delta k}{2} \geq \sum_{\alpha \in \S}t(\alpha) \geq \delta k(k-1)\\
\implies & |U|\left(k - \frac{\delta k}{2}\right) \geq \delta k^2 - \delta k - \frac{\delta k^2}{2} \implies |U| \geq \frac{\delta (k - 2)}{2 - \delta},
\end{align*}
implying $|U|\geq \frac{\delta k}{3}$ for sufficiently large $k$.

\end{proof}

\section{Our contributions}\label{sec: contribs}
In this section we give a high-level account of our contributions in this paper.
In Section~\ref{section:tsang} we discuss the PDT construction of Tsang et al. We motivate, state our results, and briefly discuss proof ideas in Sections~\ref{section:result1},~\ref{section:result2}, and~\ref{section:result3}.

\subsection{Low bucket complexity implies shallow PDTs}
\label{section:tsang}
The following lemma follows from~\cite[Lemma 28]{TWXZ13} and Equation~\eqref{eqn: 1normk}. 

\begin{lemma}[Tsang, Wong, Xie, and Zhang]
\label{lemma:tsang}
Let $f: \f^n \to \pmone$ be any Boolean function. Then there exists an affine subspace $V$ of $\f^n$ of co-dimension $O(\sqrt k)$ such that $f$ is constant on $V$.
\end{lemma}

\noindent Let $V=\bra{x \in \f^n ~\middle|~ \gamma_1(x)=b_1, \ldots, \gamma_t(x)=b_t}$ be the affine subspace $V$ obtained from Lemma~\ref{lemma:tsang}, where $t=O(\sqrt k)$. Define $\G:=\mathsf{span} \bra{\gamma_1, \ldots, \gamma_t}$. We next observe that $\mathcal B(f,\G) \leq k/2$. To see this, note that since $f \mid_V$ is constant, we have from Equation~\eqref{eqn-restr} that for each coset $C \in \mathcal C$ and any $(b_1, \dots, b_t) \in \f^t$,
\[
P_C(b_1, \ldots, b_t)=
\begin{cases}
\pm 1 & \textnormal{if}~0^n \in C\\
0 & \textnormal{otherwise}.
\end{cases}
\]
Since $f$ is a non-constant function, this implies that each $P_C(\cdot)$ has at least $2$ terms, i.e., each $\beta \in \S$ is identified with some other $\delta \in \S$ with respect to $\G$. Observation~\ref{obs:identify} implies that $\mathcal B(f,\G) \leq k/2$. Observation~\ref{obs:bucktet<sp} implies that the Fourier sparsity of the restriction of $f$ to each coset of $V$ is at most $k/2$.

This immediately leads to a recursive construction of a PDT for $f$ of depth $O(\sqrt k)$ as follows. The first step is to query the parities $\gamma_1, \ldots, \gamma_t$. After this step, each leaf of the partial tree obtained is a restriction of $f$ to some coset of $V$. Next we recursively compute each leaf. Since after each batch of queries, the sparsity reduces by a factor of $2$, the depth of the tree thus obtained is $O\left(\sqrt{k}+\sqrt{\frac{k}{2}}+\sqrt{\frac{k}{2^2}}+\cdots\right)=O(\sqrt{k})$.

\subsection{A random set of parities achieves low bucket complexity}
\label{section:result1}
Tsang et al.~proved Lemma~\ref{lemma:tsang} by an iterative procedure in each step of which a single parity is carefully chosen. We show in this paper that a randomly sampled set of parities achieves the desired bucket complexity upper bound with high probability. More specifically, for a parameter $p \in [0,1]$, consider the procedure {\sc SampleParity($f,p$)} described in Algorithm~\ref{algo1}.
\begin{algorithm}
\caption{}\label{algo1}
\begin{algorithmic}[h]
\Procedure{SampleParity ($f,p$)}{}
\State $\mathcal{R} \gets \emptyset$;
\For{each $\alpha \in \S$}
\State independently with probability $p, \mathcal R \gets \mathcal R \cup \bra{\alpha}$;
\EndFor
\State Return $\mathcal{R}$;
\EndProcedure
\end{algorithmic}
\end{algorithm}
Our first result shows that the set $\mathcal R$ returned by {\sc SampleParity$\left(f,\frac{1}{\Theta(\sqrt{k})}\right)$} satisfies $\mathcal B(f, \mathsf{span} \ \R) \leq k/2$ with high probability.
\begin{theorem}
\label{thm:thm1}
Let $f:\f^n \rightarrow \pmone$ be a Boolean function and $k$ be large enough. Let $p=\frac{1}{2k^{1/2}}$ and $\mathcal R$ be the random set of parities returned by {\sc SampleParity}($f,p$). There exists a constant $c \in [0,1)$ such that
\[
\E[\mathcal{B}(f,\mathsf{span} \ \mathcal R)] \leq ck.
\]
\end{theorem}
With high probability we have $|\R| = O(\sqrt{k})$. By an argument analogous to the discussion in the previous section, Theorem~\ref{thm:thm1} recovers the $O(\sqrt{k})$ upper bound on $\PDT(f)$.
An additional insight that our work provides is that a PDT of depth $O(\sqrt{k})$ can be obtained by a naive sampling procedure applied iteratively.

We note here that while Tsang et al.~prove a bucket complexity upper bound of $k/2$ via Lemma~\ref{lemma:tsang} which restricts the function to a constant, we derive a bucket complexity upper bound of $(1-\Omega(1))k$ by analyzing the 
procedure {\sc SampleParity}.

\paragraph*{Proof idea.}
 Fix any $\alpha \in \S$. Proposition~\ref{prop:tit} implies that for every $\beta \in \S\setminus\bra{\alpha}$, there exists $(\gamma, \delta)\in \binom{\S}{2}\setminus \bra{(\alpha, \beta)}$ such that $\alpha +\beta=\gamma+\delta$. Observe that if two parities in the set $A:=\bra{\beta, \gamma, \delta}$ are chosen in $\R$, then $\alpha$ is identified with the third parity in $A$ w.r.t.~$\mathsf{span}~\R$. Now, the expected number of $\beta \in \S\setminus\bra{\alpha}$ for which the aforementioned identification occurs is seen by linearity of expectation to be $\Omega(kp^2)$, which is $\Omega(1)$ by the choice of $p$. The crux of the proof is in strengthening this bound on expectation to conclude that with constant probability, there exists at least one $\beta \in \S\setminus \bra{\alpha}$ such that the above identification occurs. Theorem~\ref{thm:thm1} follows by linearity of expectation over $\alpha \in \S$, and an invocation of Observation~\ref{obs:identify}.

We prove Theorem~\ref{thm:thm1} in Section~\ref{sec:thm1}. In Section~\ref{sec:warmup} we prove a weaker statement that admits a simpler proof, and yet contains some key ideas that go into the proof of Theorem~\ref{thm:thm1}. \subsection{Good folding yields better PDTs}
\label{section:result2}
Assume that for any Boolean function $f$ there exist $\alpha_1, \alpha_2 \in \S$ such that $|O_{\alpha_1 + \alpha_2}| \geq k^\ell + 1$. This is a weaker assumption on $f$ than it being $(\delta,\ell)$-folding. Observation~\ref{obs:identify} implies that $\B(f,\{0^n,\alpha_1+\alpha_2\}) \leq k-k^\ell-1 \leq k(1-k^{-(1-\ell)})$. This suggests the following PDT for $f$. First the parity $\alpha_1+\alpha_2$ is queried at the root. Observation~\ref{obs:bucktet<sp} implies that the Fourier sparsity of $f$ restricted to the affine subspace (of co-dimension 1) corresponding to each outcome of this query is at most $k(1-k^{-(1-\ell)})$. Repeating this heuristic recursively for each leaf leads to a PDT of depth $O(k^{1-\ell}\log k)$.

We have now set up the backdrop to introduce our next contribution. In the preceding discussion we had assumed the following about any Boolean function $f$: there exists a pair in $\binom{\S}{2}$ with a large equivalence class. One implication of our next result is that if we instead assume that any Boolean function is $(\Omega(1),\ell)$-folding, the procedure {\sc SampleParity} with $p$ set to $1/{\widetilde{\Theta}(k^{({1+\ell)}/{2}})}$ achieves a bucket complexity upper bound of $k/2$ with high probability.
By an argument analogous to the discussion in Section~\ref{section:tsang} (also see Corollary~\ref{cor: foldingpdt}), this yields a PDT with depth $\widetilde{O}(k^{(1-\ell)/2})$. This is a quadratic improvement over the $\widetilde{O}(k^{1-\ell})$ bound discussed in the last paragraph. Besides, it can be seen to recover (up to a logarithmic factor) our first result by setting $\ell=0$, since any Boolean function is $(1,0)$-folding by Proposition~\ref{prop:tit}.

\begin{theorem}\label{thm: folding}
Let $0\leq\ell\leq1-\Omega(1)$ and $\delta\in(0,1]$. Let $f: \f^n \rightarrow \pmone$ be $(\delta,\ell)$-folding with $k$ sufficiently large. Set $p:=\frac{4000 \log k}{\delta k^{(1+\ell)/2}}$ and let $\R$ be the random subset of $\S$ that {\sc SampleParity}($f,p$) returns. Then with probability at least $1-\frac{1}{k}$, $\B(f, \mathsf{span}~\R) \leq k - \frac{\delta k}{6}$. 
\end{theorem}

The proof of Theorem~\ref{thm: folding} proceeds along the lines of that of Theorem~\ref{thm:thm1}, but is more technical. We prove it in Section~\ref{sec: improved}.

This yields the following corollary.
\begin{corollary}\label{cor: foldingpdt}
Let $0\leq\ell\leq1-\Omega(1)$ and $\delta=\Omega(1)$. Suppose all Boolean functions $f: \f^n \rightarrow \pmone$ with sufficiently large $k$ are $(\delta,\ell)$-folding.
Then,
\[
\PDT(f) = \widetilde{O}(k^{(1-\ell)/2}).
\]
\end{corollary}

\begin{proof}
Fix any Boolean function $f : \f^n \to \pmone$ with sufficiently large $k$. Let $p$ and $\R$ be as in the statement of Theorem~\ref{thm: folding}. Since $\delta$ is a constant, $p=\Theta\left(\frac{\log k}{k^{(1+\ell)/2}}\right)$. By Theorem~\ref{thm: folding}, we have $\B(f, \mathsf{span}~\R) \leq ck$, for some $c=(1-\Omega(1))$, with probability strictly greater than $1/2$. By a Chernoff bound $|\R|=\widetilde{O}(k^{(1-\ell)/2})$ with probability strictly greater than $1/2$. Finally, by a union bound, we have that with non-zero probability the set $\R$ returned by {\sc SampleParity}($f, p$) satisfies both $|\R|=\widetilde{O}(k^{(1-\ell)/2})$ and $\B(f, \mathsf{span}~\R) \leq ck$, for some $c=(1-\Omega(1))$. Choose such an $\R$ and consider the following PDT for $f$, whose construction closely follows the discussion in Section~\ref{section:tsang}.

First, query all parities in $\R$. Now, let $V$ be the affine subspace corresponding to an arbitrary leaf of this partial tree. By the properties of $\R$ and Observation~\ref{obs:bucktet<sp}, we have that the Fourier sparsity of $f\mid_V$ is at most $ck$. Repeat the same process inductively for each leaf. The depth of the resultant tree is at most $\widetilde{O}(k^{(1-\ell)/2}+(ck)^{(1-\ell)/2}+\cdots)=\widetilde{O}(k^{(1-\ell)/2})$.
\end{proof}
Corollary~\ref{cor: foldingpdt} naturally raises the question of whether all Boolean functions are $(\Omega(1), \Omega(1))$-folding.
\begin{question}
\label{qn:foldingquest1}
Do there exist constants $\ell, \delta \in (0,1]$ such that every Boolean function $f:\f^n \rightarrow \pmone$ is $(\delta, \ell)$-folding?
\end{question}
An affirmative answer to Question~\ref{qn:foldingquest1} in conjunction with Corollary~\ref{cor: foldingpdt} and the discussion in Section~\ref{sec:introduction} implies an upper bound on the communication complexity of XOR functions $F = f \circ \oplus$ that is polynomially smaller than the best known bound of $O(\sqrt{\mathsf{rank}(F)})$.

What is the largest $\ell$ for which all Boolean functions are $(\Omega(1),\ell)$-folding? The addressing function $\mathsf{ADD}_k$ (see Definition~\ref{def:addressing}) is $(1,1/2-o(1))$-folding, and not $(\Omega(1),\ell)$-folding for any $\ell \geq \frac{1}{2}$ (see Appendix~\ref{sec: addfold}). In light of this, we make the following conjecture.
\begin{conjecture}\label{conj: folding}
There exists a constant $\delta > 0$ such that any Boolean function $f : \f^n \to \pmone$ is $(\delta, 1/2-o(1))$-folding.
\end{conjecture}
Assuming Conjecture~\ref{conj: folding}, Corollary~\ref{cor: foldingpdt} would imply an upper bound of $\widetilde{O}(\mathsf{rank}^{1/4+o(1)}(F))$ on the communication complexity of XOR functions $F = f \circ \oplus$.

\subsection{Boolean functions have non-trivial folding properties}
\label{section:result3}
Recall that Conjecture~\ref{conj: folding} states that any Boolean function is $(\delta, \ell)$-folding with $\delta = \Omega(1)$ and $\ell = 1/2 - o(1)$.
Also recall from Proposition~\ref{prop:tit} that a necessary condition for a function to be Boolean valued is that it is $(\delta, \ell)$-folding with $\delta = 1$ and $\ell = 0$. We show in the appendix (see Theorem~\ref{thm: ceg}) that the conditions in Proposition~\ref{prop:tit} are not sufficient for a function to be Boolean valued.

To the best of our knowledge, it was not known prior to our work whether any better bound than this was known for Boolean functions (in terms of $\ell$, for \emph{any} non-zero $\delta$).
In particular, it was consistent with prior knowledge that there exist functions for which each equivalence class of $\binom{\S}{2}$ contains exactly 2 elements.
We rule out this possibility, and our contribution is a step towards Conjecture~\ref{conj: folding}.

\begin{theorem}\label{thm: 3foldintro}
For any Boolean function $f : \f^n \to \pmone$ with $k > 4$, and every $\alpha \in \S$, there exists $\beta \in \S \setminus \bra{\alpha}$ such that $|O_{\alpha + \beta}| \geq 3$.
\end{theorem}
In order to rule out the possibility mentioned above, it suffices to exhibit a single pair $(\alpha, \beta) \in \binom{\S}{2}$ with $|O_{\alpha + \beta}| \geq 3$. Theorem~\ref{thm: 3foldintro} further shows that \emph{every} element $\alpha \in \S$ participates in such a pair.
\paragraph*{Proof idea} We prove this via a series of arguments. Define $\S_+:=\bra{\alpha \in \S ~\middle|~ \wh{f}(\alpha) > 0}$ and $\S_-:=\bra{\alpha \in \S ~\middle|~ \wh{f}(\alpha) < 0}$. We first show that if there exists $\alpha \in \S$ with $|O_{\alpha + \beta}| = 2$ for all $\beta \in \S \setminus \bra{\alpha}$, then both of the following hold.
\begin{enumerate}
    \item Either $|\S_+|$ or $|\S_-|$ is odd.
    \item The function $f$ must be plateaued.
\end{enumerate}
The proofs use Equation~\eqref{eqn: boolstruct}.
Next, we show that for plateaued Boolean functions, both $|\S_+|$ and $|\S_-|$ are even, yielding a contradiction in view of the first bullet above.  This proof involves a careful analysis of the Fourier coefficients and crucially uses Observation~\ref{obs: gran} and Equation~\eqref{eqn: parseval}.

A natural question raised by Theorem~\ref{thm: 3foldintro} is whether there exists a Boolean function $f$ and $\alpha \in \S$ such that there exists only one element $\beta \in \S \setminus \bra{\alpha}$ with $|O_{\alpha + \beta}| \geq 3$. The following theorem answers this question in the positive, and sheds more light on the structure of such functions. 

\begin{theorem}\label{thm: singledirectionintro}
~
\begin{enumerate}
\item \label{item: existsfunction} There exists a Boolean function $f : \f^n \to \pmone$ and $(\alpha, \beta) \in \binom{\S}{2}$ such that $|O_{\alpha + \gamma}| = 2$ for all $\gamma \in \S \setminus \bra{\alpha, \beta}$.
\item \label{item: allfunctions} Let $f : \f^n \to \pmone$ be any Boolean function. If there exists $(\alpha, \beta) \in \binom{S}{2}$ such that $|O_{\alpha + \gamma}| = 2$ for all $\gamma \in \S \setminus \bra{\alpha, \beta}$, then $|O_{\alpha + \beta}| = k/2$.
\end{enumerate}
\end{theorem}
The proof of Part~\ref{item: allfunctions} of Theorem~\ref{thm: singledirectionintro} follows along the lines of the proof of Theorem~\ref{thm: 3foldintro}. The proof of Part~\ref{item: existsfunction} of Theorem~\ref{thm: singledirectionintro} constructs such a function by applying a simple modification to the addressing function.

We prove Theorems~\ref{thm: 3foldintro} and~\ref{thm: singledirectionintro} in Section~\ref{sec: nontrivial}.

\section{Proof of Theorem~\ref{thm:thm1}}
\label{sec:firstthm}

In this section we  prove our first result, Theorem~\ref{thm:thm1}.

\subsection{Warm up: sampling $\widetilde{O}(k^{3/4})$ parities.}
\label{sec:warmup}
In this section we prove a quantitatively weaker statement. This admits a simpler proof and introduces many key ideas that go into our proof of Theorem~\ref{thm:thm1}.
\begin{claim}
\label{clm:klemone}
Let $p:=\frac{2\sqrt{\log k}}{k^{1/4}}$, and let $\R$ be the set returned by {\sc SampleParity}($f,p$). Then
\[
\Pr[\B(f, \mbox{span }\R) \leq k/2] \geq 1-\frac{1}{k^{1/3}}.
\]
\end{claim}
By a Chernoff bound, with high probability, $|\R|=\widetilde{O}(k^{3/4})$.
\begin{proof}
Fix any $\alpha \in \S$. By Proposition~\ref{prop:tit} we have that for each $\beta \in \S \setminus\bra{\alpha}$, there exist $\beta_1, \beta_2 \in \S\setminus \bra{\alpha,\beta}$ such that $\alpha+\beta+\beta_1+\beta_2=0$. Define $Q_\beta:=\bra{\beta, \beta_1, \beta_2}$. Note that the sets $Q_\beta$ are not necessarily distinct. Define the multiset of unordered triples $\mathcal{F}:=\bra{Q_\beta ~\middle|~ \beta \in \S \setminus\bra{\alpha}}$. For each $\gamma \in \S\setminus\bra{\alpha}$, define $\mathcal{D}_\gamma:=\bra{\beta \in \S\setminus\bra{\alpha} ~\middle|~ \gamma \in Q_\beta}$. We now show that with high probability there exists $F \in \mathcal{F}$ such that $|F \cap \R| \geq 2$. We consider two cases below.

\begin{description}
\item[Case 1: There exists $\gamma \in \S\setminus\bra{\alpha}$ such that $|\mathcal{D}_\gamma| \geq k^{1/2}$.]\ \\
Consider the multiset of unordered pairs $\mathcal{A}:=\bra{Q_\beta\setminus\bra{\gamma} ~\middle|~ \beta \in \mathcal{D}_\gamma}$. Each pair in $\mathcal{A}$ can repeat at most thrice. Hence there are at least $k^{1/2}/3$ distinct pairs in $\mathcal{A}$. Moreover the distinct pairs in $\mathcal{A}$ are disjoint. This can be inferred from the observation that the sum of the two elements in each pair in $\mathcal{A}$ equals $\alpha+\gamma$. Thus 
\[
\Pr\left[\forall A \in \mathcal{A}, A \nsubseteq \R\right] \leq (1-p^2)^{k^{1/2}/3} = \left(1-\frac{4 \log k}{k^{1/2}}\right)^{k^{1/2}/3} \leq \frac{1}{k^{4/3}}.
\]
\item[Case 2: For each $\gamma \in \S \setminus \bra{\alpha}$, $|\mathcal{D}_\gamma| < k^{1/2}$.]\ \\
In this case each triple in $\mathcal{F}$ has non-empty intersection with at most $3k^{1/2}$ sets in $\mathcal{F}$. Thus one can greedily obtain a collection $\mathcal{T}$ of at least $\frac{k-1}{3k^{1/2}}$ disjoint triples in $\mathcal{F}$. 
\[
\Pr\left[\forall T \in \mathcal{T}, |T \cap \R| < 2 \right] \leq (1-p^2)^{\frac{k-1}{3k^{1/2}}}=\left(1-\frac{4 \log k}{k^{1/2}}\right)^{\frac{k-1}{3k^{1/2}}},
\]
which is at most $\frac{1}{k^{4/3}}$ for large enough $k$.
\end{description}
From the above two cases it follows that with probability at least $1-\frac{1}{k^{4/3}}$, there exists a triple $F \in \mathcal{F}$ such that $|F \cap \R| \geq 2$. Assume existence of such a triple $F$, and let $\delta_1, \delta_2 \in F \cap \R$. Let $\delta:=F \setminus \{\delta_1, \delta_2\}$. Since $\alpha+\delta_1+\delta_2+\delta=0^n$, we have that $\alpha +\delta=\delta_1+\delta_2 \in \mathsf{span}~\R$, i.e., $\alpha$ is identified with $\delta$ with respect to $\mathsf{span}~\R$.  By a union bound over all $\alpha \in \S$ it follows that with probability at least $1-\frac{1}{k^{1/3}}$, for every $\alpha \in \S$ there exists a $\delta \in \S \setminus\bra{\alpha}$ such that $\alpha$ is identified with $\delta$ w.r.t.~$\R$. The claim follows by Observation~\ref{obs:identify}.
\end{proof}

\subsection{Sampling $O(k^{1/2})$ parities}
\label{sec:thm1}
We now proceed to prove Theorem~\ref{thm:thm1} by refining the ideas developed in Section~\ref{sec:warmup}. Recall that by a Chernoff bound, $|\R| = O(\sqrt{k})$ with high probability (where $\R$ is as in Theorem~\ref{thm:thm1}).
We require the following inequality.
\begin{proposition}\label{prop: calculus}
For any non-negative integer $d$, and $p \in [0, 1]$ be such that $pd \leq 1$. Then,
\[(1-p)^d \leq 1-\frac{1}{2}pd.\]
\end{proposition}
\begin{proof}
The proof proceeds via induction on $d$.
\begin{description}
\item[Base case: d=0.] The statement can be easily verified to be true; each side evaluates to 1.
\item[Inductive step: ] Assume that the statement is true for $d\geq 0$ and all $p \in [0, 1]$ such that $pd \leq 1$. We now show that the hypothesis holds for $d+1$ and all $p \in \left[0, \frac{1}{d + 1}\right]$. We have 
\begin{align*}
(1-p)^{d+1}=&(1-p)\cdot(1-p)^d \\
\leq&(1-p)\left(1-\frac{1}{2}pd\right)\tag*{by inductive hypothesis, since $pd \leq p(d+1) \leq 1$} \\
=&1-\left(\frac{1}{2}p+\frac{1}{2}pd\right)-\frac{1}{2}p +\frac{1}{2}p^2d \\
\leq&1-\frac{1}{2}p(d+1).\tag*{since $pd \leq 1$}
\end{align*}
\end{description}
\end{proof}
\begin{proof}[Proof of Theorem~\ref{thm:thm1}] For technical reasons we instead consider a two-step probabilistic procedure. Define $p':=\frac{1}{4k^{1/2}}$. Let $\R_1$ and $\R_2$ be the sets returned by two independent runs of {\sc SampleParity}($f, p'$), and let $\R' := \R_1 \cup \R_2$. Each $\alpha \in \S$ is independently included in $\R'$ with probability equal to $1 - (1 - p')^2  < 2p' = p$.
Hence it suffices to prove that there exists a constant $c \in (0,1]$ such that $\E[\B(f, \mathsf{span}~\R')] \leq ck$.

Fix any $\alpha \in \S$ and let $Q_\beta$ and $\mathcal{F}$ be as in the proof of Claim~\ref{clm:klemone}. For $\gamma \in \S\setminus\bra{\alpha}$, define $\wdeg(\gamma):=|\bra{\beta \in \S\setminus\bra{\alpha}~\middle|~\gamma \in Q_{\beta}\setminus \bra{\beta}}|$. Clearly, $\E_{\gamma \sim \S\setminus\bra{\alpha}}[\wdeg(\gamma)]=2$.
Define $A:=\bra{\gamma \in \S \setminus \bra{\alpha}~\middle|~\wdeg(\gamma)\geq 4k^{1/2}}$. By Markov's inequality, $|A| \leq k^{1/2}/2$. Fix an ordering $\sigma$ on $\S\setminus\bra{\alpha}$ such that all elements of $\overline{A}:=(\S \setminus \bra{\alpha}) \setminus A$  appear before all elements of $A$.

Define $T:=\bra{\beta \in \S\setminus\bra{\alpha} ~\middle|~ Q_\beta \setminus \bra{\beta} \subseteq A}$. Observe that the pairs $Q_\beta \setminus\bra{\beta}$ for distinct $\beta \in \S\setminus\bra{\alpha}$ are distinct. This can be inferred from the observation that the sum (with respect to coordinate-wise addition in $\f$) of the two elements of $Q_\beta \setminus\bra{\beta}$ equals $\alpha+\beta$. This gives us the following bound on the size of $T$:
\begin{equation}\label{eqn: tupperbound}
|T| \leq \binom{|A|}{2} \leq \frac{k}{8}.
\end{equation}
Define $\overline{T}:=(\S\setminus\bra{\alpha})\setminus T$. For each $\beta \in \overline{T}$, the first character (according to $\sigma$) in the pair $Q_\beta \setminus \bra{\beta}$ is from $\overline{A}$. For each $\gamma \in \overline{A}$, define $\d(\gamma)$ to be the number of $\beta \in \overline{T}$ such that $\gamma$ is the first element in $Q_\beta \setminus \bra{\beta}$. By Equation~\eqref{eqn: tupperbound} we have
\begin{align}\sum_{\gamma \in \overline{A}} \d(\gamma)=|\overline{T}| \geq k-1- \frac{k}{8}\geq \frac{2k}{3}\label{degreebound}\end{align}
where the last inequality holds for large enough $k$.

For $\gamma \in \overline{A}$, let $\mathcal E(\gamma)$ be the event that there exists $\beta \in \overline{T} \cap \R_1$ such that $\gamma$ is the first element in $Q_\beta \setminus \bra{\beta}$. We have
\begin{equation}\label{prob-bound}
\Pr_{\R_1}[\mathcal E(\gamma)] = 1-(1-p')^{\d(\gamma)} \geq \frac{p'\cdot\d(\gamma)}{2},
\end{equation}
where the last inequality follows by Proposition~\ref{prop: calculus}.
Here Proposition~\ref{prop: calculus} is applicable since $\d(\gamma) \leq \wdeg(\gamma) \leq 4k^{1/2}$ (since $\gamma \in \overline{A}$), and $p' = \frac{1}{4k^{1/2}}$.
Define the random set $B:=\bra{\gamma \in \overline{A} ~\middle|~ \mathcal{E}(\gamma)\mbox{ occurs}}$. We have
\begin{align*}
\E_{\R_1}[|B|] & = \sum_{\gamma \in \overline{A}}\Pr_{\R_1}[\mathcal{E}(\gamma)] \geq \sum_{\gamma \in \overline{A}} \frac{p' \cdot \d(\gamma)}{2} \tag*{by linearity of expectation and Equation~\eqref{prob-bound}}\\
& \geq \frac{1}{2} \cdot \frac{1}{4k^{1/2}} \cdot \frac{2k}{3} \geq \frac{k^{1/2}}{12}. \tag*{by Equation~\eqref{degreebound}, and substituting the value of $p'$}
\end{align*}
Furthermore, the events $\mathcal E(\gamma)$ are independent. By a Chernoff bound, 
$\Pr_{\R_1}\left[|B|\geq \frac{k^{1/2}}{24}\right] \geq 0.9$. 
Now,
\begin{align*}
\Pr_{\R_1, \R_2} \left[B \cap \R_2\neq\emptyset ~\middle|~ |B| \geq \frac{k^{1/2}}{24}\right] & \geq 1-(1-p')^{k^{1/2}/24} \geq 1-e^{-p'\cdot \frac{k^{1/2}}{24}}=1-e^{-\frac{1}{96}}\\
& = c_1,~\text{say}.
\end{align*}
Thus, the probability of the event $\mathcal{E}:=\bra{|B|\geq \frac{k^{1/2}}{24}} \wedge \bra{B \cap \R_2\neq\emptyset}$ is at least $0.9 c_1$. Suppose the event $\mathcal{E}$ occurs, and let $\gamma \in B \cap \R_2$. By the definitions of $B$ and $\mathcal{E}(\gamma)$, there exists $\beta \in \overline{T}\cap\R_1$ such that $\gamma$ is the first element of $Q_{\beta}\setminus\bra{\beta}$. Let $\delta:=Q_{\beta}\setminus\bra{\beta, \gamma}$. Then, $\alpha+\delta=\beta+\gamma$. Since $\beta \in \R_1$ and $\gamma \in \R_2$, $\alpha$ is identified with $\delta$ with respect to $\mathsf{span}~\R'$. 
In summary, we have shown that for any $\alpha \in \S$,
\[
\Pr_{\R_1, \R_2}[\alpha~\text{is identified with some}~\delta \in \S \setminus \bra{\alpha}~\text{w.r.t.}~\mathsf{span}~\R'] \geq 0.9 c_1.
\]
By linearity of expectation,
\[
\E_{\R_1, \R_2}[\abs{\bra{\alpha \in \S~\middle|~\alpha~\text{is identified with some}~\delta \in \S \setminus \bra{\alpha}~\text{w.r.t.}~\mathsf{span}~\R'}}] \geq k\cdot 0.9c_1.
\]
Observation~\ref{obs:identify} then implies
\[
\E_{\R_1, \R_2}[\B(f, \mathsf{span}~\R')] \leq k - \frac{k\cdot 0.9 c_1}{2} = ck,
\]
where $c = \left(1 - \frac{0.9 c_1}{2}\right)$.
\end{proof}

\section{Proof of Theorem~\ref{thm: folding}}\label{sec: improved}

In this section we prove Theorem~\ref{thm: folding}. This theorem gives us a tool to construct efficient PDTs if all Boolean functions have ``good folding properties''.

\begin{proof}[Proof of Theorem~\ref{thm: folding}]
Since $f$ is $(\delta, \ell)$-folding, Claim~\ref{claim:folding} implies that many $\alpha \in \S$ participate in many large folding directions. More precisely, if we define
\[
U:=\bra{\alpha \in \S ~\middle|~ \mbox{there exist at least $\delta k/2$ many $\beta \in \S \setminus \bra{\alpha}$ with $|O_{\alpha+\beta}|\geq k^\ell+1$}},
\]
then $|U| \geq \delta k /3$.

As in the proof of Theorem~\ref{thm:thm1}, for technical reasons we consider a slightly different probabilistic procedure than {\sc SampleParity}($f, p$). Define $p_1:=\frac{4\log k}{5e\delta k^{(1+\ell)/2}}$ and $p_2:=\frac{2000 \log k}{\delta k^{(1+\ell)/2}}$. Let $\R_1$ and $\R_2$ be the sets returned by independent runs of {\sc SampleParity}($f, p_1$) and {\sc SampleParity}($f, p_2$), respectively, and let $\R' := \R_1 \cup \R_2$. Each $\alpha \in \S$ is independently included in $\R'$ with probability equal to $1 - (1 - p_1)(1 - p_2) = p_1 + p_2 - p_1p_2 < 2p_2 = p$.
Hence it suffices to prove that with probability at least $1 - \frac{1}{k}$, we have $\B(f, \mathsf{span}~\R') \leq k - \frac{\delta k}{6}$.

We now show that with probability at least $1 - \frac{1}{k}$, each element of $U$ is identified with some other element of $\S$ with respect to $\mathsf{span}\ \R'$. The theorem would then follow by Observation~\ref{obs:identify}.
To this end, fix an $\alpha \in U$. Define $T$ to be a set of arbitrarily chosen $\delta k/2$ elements $\beta \in \S \setminus \bra{\alpha}$ such that $|O_{\alpha+\beta}|\geq k^\ell+1$. 
For each $\beta \in T$, there are at least $k^\ell$ distinct pairs $(\beta_1,\beta_2)\in O_{\alpha + \beta} \setminus \bra{(\alpha,\beta)}$. Let the set of any $k^\ell$ such pairs be $P_{\beta}$. Define $P:=\bigcup_{\beta \in T} P_{\beta}$. Since the $P_{\beta}$'s are pairwise disjoint,
\begin{equation}\label{eqn: psize}
|P|=\delta k^{1+\ell}/2.
\end{equation}
For each $\gamma \in \S$, define $\wdeg(\gamma)$ to be the total number of pairs in $P$ that $\gamma$ appears in. By a counting argument, Equation~\eqref{eqn: psize} implies
\[
\E_{\gamma \in \S}[\wdeg(\gamma)] = \delta k^\ell.
\]
Define $A:=\bra{\gamma \in \S ~\middle|~ \wdeg(\gamma) > 5 k^{(1+\ell)/2}/2}$. 
By Markov's inequality, 
\begin{equation}\label{eqn: aupperbd}
|A| < \delta k ^\ell \frac{2}{5 k^{(1+\ell)/2}}\cdot k= \frac{2\delta k^{(1+\ell)/2}}{5}.
\end{equation}
Define $\overline{A}:=(\S \setminus\bra{\alpha})\setminus A$.
Define 
\[
P' := \bra{(\mu, \nu) \in P ~\middle|~ \mu \in \overline{A}~\text{or}~\nu \in \overline{A}}.
\]
By Equations~\eqref{eqn: psize} and~\eqref{eqn: aupperbd}, 
\begin{equation}\label{eqn: p'lb}
|P'| \geq \frac{\delta k^{1+\ell}}{2}-{|A| \choose  2}\geq \frac{\delta k^{1+\ell}}{2}-\frac{4\delta^2 k^{1+\ell}}{50} = \delta k^{1 + \ell}\left(\frac{1}{2} - \frac{2\delta}{25}\right)> \frac{\delta k^{1+\ell}}{4}.
\end{equation}
Define $T':=\bra{\beta \in T ~\middle|~ |P_{\beta} \cap P'| \geq k^\ell/8}$. Next we show that $T'$ has $\Omega(\delta k)$ elements. By the definition of $T'$, we have
\[
|P'| \leq |T'|\cdot k^\ell +\left(\frac{\delta k}{2}-|T'|\right)\cdot\frac{k^\ell}{8}.
\]
Along with Equation~\eqref{eqn: p'lb}, we obtain
\begin{align}
&|T'|\cdot k^\ell +\left(\frac{\delta k}{2}-|T'|\right)\cdot\frac{k^\ell}{8}\geq \frac{\delta k^{1+\ell}}{4}\nonumber\\
\implies &|T'|\cdot\frac{7k^\ell}{8} \geq \frac{3\delta k^{1+\ell}}{16}\nonumber\\
\implies &|T'|\geq \frac{3\delta k}{14}. \label{eqn: t'lb}
\end{align}
For $\gamma \in \overline{A}$ and $\beta \in T'$ we say that \emph{$\beta$ hits $\gamma$} if $\gamma$ appears in a pair in $P_{\beta}$.
For any $\gamma \in \overline{A}$, define 
\[
\d(\gamma) := |\bra{\beta \in T' ~\middle|~ \beta~\text{hits}~\gamma}|.
\]
Observe that $\gamma$ can appear in at most one pair in any $P_\beta$. For a fixed $\gamma \in \overline{A}$, let $\mathcal{E}(\gamma)$ denote the following event.
\[
\mathcal{E}(\gamma) := \bra{\lvert\bra{\beta \in T' ~\middle|~ \beta \in \R_1 \textnormal{ and } \beta \textnormal{ hits } \gamma}\rvert \geq \frac{4 \log k}{\delta}}.
\]
We have $\d(\gamma) \leq \wdeg(\gamma) \leq 5k^{(1 + \ell)/2}/2$, where the last inequality follows from the definition of $\overline{A}$.
Thus, for any fixed $\gamma \in \overline{A}$ we have
\begin{align*}
\Pr_{\R_1}[\mathcal{E}(\gamma)] & \leq {\binom{\d(\gamma)} {(4 \log k)/\delta}}p_1^{(4\log k)/\delta} \leq \left(\frac{e\delta\cdot\d(\gamma) \cdot p_1}{4\log k}\right)^{(4\log k)/\delta}\\
& \leq \left(\frac{1}{2}\right)^{(4\log k)/\delta} \tag*{since $p_1 = \frac{4\log k}{5e\delta k^{(1+\ell)/2}}$}\\
&\leq \frac{1}{3k^3}.\tag*{for large enough $k$}
\end{align*}

Since $|\overline{A}| \leq k$, we have by a union bound that 
\begin{align}\label{eqn: egammaub}
\Pr_{\R_1}\left[\bigcup_{\gamma \in \overline{A}}\mathcal{E}(\gamma)\right] \leq \frac{1}{3k^2}.
\end{align}
Recall from Equation~\eqref{eqn: t'lb} that $|T'| \geq 3\delta k/14$. By our choice of $p_1$, the expected number of elements in $T'$ that are included in $\R_1$ is at least $\frac{3\delta k}{14} \cdot \frac{4\log k}{5e\delta k^{(1+\ell)/2}} = \frac{6}{35e} \cdot k^{(1 - \ell)/2}\log k$. By a Chernoff bound, the number of elements $\beta$ in $T'$ that are included in $\R_1$ is at least $\frac{1}{25}\cdot k^{(1-\ell)/2}\log k$ with probability at least $1 - \exp(-\Omega(k^{(1-\ell)/2}\log k))$. Since $\ell \leq 1-\Omega(1)$, this probability is at least $1 - \exp(-(k^{\Omega(1)}))$, which is at least $1-\frac{1}{3k^2}$ for large enough $k$. Define $X := \R_1 \cap T'$ and $\mathcal{F}$ to be the event $\bra{\overline{\bigcup_{\gamma \in \overline{A}}\mathcal{E}(\gamma)}} \wedge \bra{\abs{X} > \frac{1}{25}\cdot k^{(1-\ell)/2}\log k}$. By Equation~\eqref{eqn: egammaub} and a union bound,
\[
\Pr_{\R_1}[\mathcal{F}] > 1 - \frac{2}{3k^2}.
\]
Define $Y:= \bra{\gamma \in \overline{A} ~\middle|~ \gamma~\text{is hit by some}~\beta \in X}$. Now condition on $\mathcal{F}$.
By the definitions of $T'$ and $\mathcal{E}(\gamma)$, and from the fact that $|X|\geq\frac{1}{25}\cdot k^{(1-\ell)/2}\log k$ under the above conditioning, it follows that
\[
\abs{Y} \geq \frac{|X|\cdot(k^\ell/8)}{(4\log k)/\delta} \geq \frac{\delta k^{(1+\ell)/2}}{800}.
\]
Next we proceed to the second phase of sampling. We have
\[
\Pr_{\R_1, \R_2}\left[\R_2 \cap Y \neq \emptyset~\middle|~\mathcal{F} \right] \geq 1-(1-p_2)^{\frac{\delta k^{(1+\ell)/2}}{800}} \geq 1 - e^{-p_2\cdot\frac{\delta k^{(1+\ell)/2}}{800}} \geq 1-\frac{1}{k^3}.
\]
Condition on the event that $\R_2 \cap Y \neq \emptyset$ and let $\gamma \in \R_2 \cap Y$. By the definition of $Y$, we have that $\gamma$ hits some $\beta \in \R_1 \cap T' \subseteq \R_1$, i.e., $\gamma$ appears in a pair in $P_\beta$, say $(\gamma, \mu)$. Thus, $\alpha + \beta = \gamma + \mu$, which implies $\alpha + \mu = \beta + \gamma$. Since $\beta \in \R_1$ and $\gamma \in \R_2$, we have $\alpha + \mu  \in \mathsf{span}~\R'$.
We have thus shown that conditioned on the events $\mathcal{F}$ and $\R_2 \cap Y \neq \emptyset$, $\alpha$ is identified with $\mu$ with respect to $\mathsf{span}~\R'$ with probability $1$. By a union bound, $\Pr_{\R_1, \R_2}[\mathcal{F} \cap \bra{\R_2 \cap Y \neq \emptyset}] > 1-\frac{1}{k^2}$ for large enough $k$. The theorem follows by a union bound over all $\alpha \in U$.
\end{proof}

\section{Proofs of Theorem~\ref{thm: 3foldintro} and Theorem~\ref{thm: singledirectionintro}}\label{sec: nontrivial}

In this section we prove Theorem~\ref{thm: 3foldintro} and Theorem~\ref{thm: singledirectionintro}. Theorem~\ref{thm: 3foldintro} states that for any Boolean function $f : \f^n \to \pmone$ and $\alpha \in \S$, there exists at least one $\beta \in \S$ with $|O_{\alpha + \beta}| \geq 3$. Theorem~\ref{thm: singledirectionintro} consists of two parts; the first part asserts existence of a function $f : \f^n \to \pmone$ and $\alpha \in \S$ for which there exists only one $\beta \in \S$ with $|O_{\alpha + \beta}| \geq 3$, and the second part gives additional structure on such functions. 

We first recall and introduce some notation. 
Recall from Proposition~\ref{prop:tit} that for any Boolean function $f : \f^n \to \pmone$ and every $\gamma \in (\S + \S) \setminus \bra{0^n}$, we have $|O_{\gamma}| \geq 2$. For any $\gamma$ with $|O_{\gamma}| > 2$, we say that $\gamma$ is a \emph{non-trivial folding direction}. Hence, Theorem~\ref{thm: 3foldintro} can be rephrased to say that for any Boolean function $f : \f^n \to \pmone$, every element $\alpha \in \S$ must participate in at least one non-trivial folding direction.
For any Boolean function $f : \f^n \to \pmone$, define $\S_+ := \bra{\alpha \in \S ~\middle|~ \wh{f}(\alpha) > 0}$, and $\S_- := \bra{\alpha \in \S ~\middle|~ \wh{f}(\alpha) < 0}$. For any set $S$, we use the notation $\binom{S}{3}$ to denote the set of all subsets of $S$ of size exactly $3$. We abuse notation and denote a generic element of $\binom{S}{3}$ as $(a, b, c)$ rather than $\bra{a, b, c}$.

We require the following proposition.
\begin{proposition}
\label{prop:wlog}
Let $f:\f^n \to \pmone$ be a Boolean function with Fourier support $\S$ with $k = |\S| \geq 2$. Let $\alpha, \beta$ be two distinct parities in $\S$. Then, there exists a Boolean function $g: \f^n \to \pmone$ with Fourier support $\S$ and $\wh{g}(\alpha)>0,\wh{g}(\beta)>0$.
\end{proposition}
\begin{proof}
If $\wh{f}(\alpha)\wh{f}(\beta) >0$, then the proposition follows by setting $g = f$ if $\wh{f}(\alpha)>0$, and $g = -f$ otherwise. Hence we may assume that $\wh{f}(\alpha)\wh{f}(\beta)<0$. Fix any $y\in\f^n$ such that $\chi_{\alpha+\beta}(y)=-1$, i.e., $\chi_\alpha(y)=-\chi_\beta(y)$. Define $h(x):=f(x+y)$. Then,
\[h(x)=\sum_{\delta \in \f^n}\wh{f}(\delta)\chi_\delta(x + y) = \sum_{\delta \in \f^n}\wh{f}(\delta)\chi_\delta(y)\cdot\chi_\delta(x),\]
giving us
$\wh{h}(\delta)=\wh{f}(\delta)\cdot\chi_{\delta}(y)$ for all $\delta \in \f^n$. In particular,
$\wh{h}(\alpha)=\wh{f}(\alpha)\cdot\chi_{\alpha}(y)$ and $\wh{h}(\beta)=\wh{f}(\beta)\cdot\chi_{\beta}(y)$. By the choice of $y$ and the assumption $\wh{f}(\alpha)\wh{f}(\beta)<0$ we have that $\wh{h}(\alpha)\wh{h}(\beta) = \wh{f}(\alpha)\wh{f}(\beta) \chi_{\alpha + \beta}(y)>0$. Then as described before, $g$ can be taken to be $h$ or $-h$ depending on the sign of $\wh{h}(\alpha)$. Finally, note that for each $\delta \in \f^n$, $|\wh{g}(\delta)|=|\wh{f}(\delta)|$, implying that the Fourier support of $g$ is $\S$.
\end{proof}

We next state a preliminary claim.

\begin{claim}\label{claim: trivialfoldodd}
Let $f : \f^n \to \pmone$ be any Boolean function. Suppose there exists $\alpha \in \S$ such that $|O_{\alpha + \beta}| = 2$ for all $\beta \in \S \setminus \bra{\alpha}$. Then, either $|\S_+|$ is odd or $|\S_-|$ is odd.
\end{claim}

\begin{proof}
Fix any set $\alpha \in \S$ such that $|O_{\alpha + \beta}| = 2$ for all $\beta \in \S \setminus \bra{\alpha}$. Assume $\alpha \in \S_+$ (else run this argument with $\S_+$ and $\S_-$ interchanged).
Consider the set of unordered triples
\[
T = \bra{(\beta, \gamma, \delta) \in \binom{\S \setminus \bra{\alpha}}{3} ~\middle|~ \alpha + \beta + \gamma + \delta = \emptyset}.
\]
Let $T_{+}$ denote the set of triples in $T$ that contain at least one element $\beta \in \S_+$, i.e.,
\[
T_+ := \bra{(\beta, \gamma, \delta) \in T ~\middle| ~\textnormal{at least one of}~ \wh{f}(\beta), \wh{f}(\gamma) , \wh{f}(\delta) ~\textnormal{is positive}}.
\]

Since $|O_{\alpha + \beta}| = 2$ for all $\beta \in \S \setminus \bra{\alpha}$, this implies that any $\beta \in \S$ (in particular any $\beta \in \S_+$) appears in exactly one triple. For any $\beta \in \S_+$, say this triple is $(\beta, \beta_1, \beta_2)$. Equation~\eqref{eqn: boolstruct} implies that
\[
\wh{f}(\alpha)\wh{f}(\beta) + \wh{f}(\beta_1)\wh{f}(\beta_2) = 0.
\]
Since $\alpha$ and $\beta$ are both in $\S_+$, exactly one of $\beta_1, \beta_2$ is in $\S_+$ and the other is in $\S_-$. 

Thus each triple in $T_+$ contains exactly two elements of $\S_+$, and none of these elements appears in any other triple. Moreover each element of $\S_+$ appears in some triple in $T_+$. Accounting for $\alpha$ being in $\S_+$, we conclude that if $|T_+| = t$, then $|\S_+| = 2t + 1$, which is odd.
\end{proof}
We state another claim that we require.
\begin{claim}\label{claim: 2foldingplat}
Let $f : \f^n \to \pmone$ be any Boolean function. If there exists $\alpha \in \S$ such that $|O_{\alpha + \beta}| = 2$ for all $\beta \in \S \setminus \bra{\alpha}$, then $f$ is plateaued.
\end{claim}

\begin{proof}[Proof of Claim~\ref{claim: 2foldingplat}]
Fix any $\alpha \in \S$ such that $|O_{\alpha + \beta}| = 2$ for all $\beta \in \S \setminus \bra{\alpha}$. Towards a contradiction, suppose $f$ is not plateaued. This implies existence of $\gamma \in \S$ such that  $|\wh{f}(\alpha)| \neq |\wh{f}(\gamma)|$. Proposition~\ref{prop:tit} implies existence of $\mu, \nu \in \S$ be such that $\alpha + \gamma = \mu + \nu$. We also have that
\[
\alpha + \nu = \mu + \gamma, \qquad  \alpha + \mu = \gamma + \nu.
\]
Arrange $\alpha, \gamma, \mu$ and $\nu$ in non-increasing order of the absolute values of their Fourier coefficients. Let the resultant sequence be $\delta_1, \delta_2, \delta_3, \delta_4$. Thus,
\[
|\wh{f}(\delta_1)| \geq |\wh{f}(\delta_2)| \geq |\wh{f}(\delta_3)| \geq |\wh{f}(\delta_4)|.
\]
Since $|\wh{f}(\alpha)| \neq |\wh{f}(\gamma)|$, at least one of these inequalities must be strict, which in particular implies that $|\wh{f}(\delta_1)||\wh{f}(\delta_2)| > |\wh{f}(\delta_3)||\wh{f}(\delta_4)|$. Now by the hypothesis, for all $1 \leq i < j \leq 4$, and $\bra{k,m}:=\bra{1,2,3,4}\setminus\bra{i,j}$ we have that $|O_{\delta_i+\delta_j}|=|O_{\delta_k+\delta_m}| = 2$. Thus, by Equation~\eqref{eqn: boolstruct} we have that $\wh{f}(\delta_1)\wh{f}(\delta_2) = - \wh{f}(\delta_3)\wh{f}(\delta_4)$, implying that $|\wh{f}(\delta_1)||\wh{f}(\delta_2)| = |\wh{f}(\delta_3)||\wh{f}(\delta_4)|$, which is a contradiction.
\end{proof}
The next claim shows that Theorem~\ref{thm: 3foldintro} holds true if $f$ is a plateaued function.
\begin{claim}\label{claim: platfold}
Let $f : \f^n \to \pmone$ be any plateaued Boolean function with $k > 4$. Then, for any $\alpha \in \S$, there exists $\beta \in \S \setminus \bra{\alpha}$ such that $|O_{\alpha + \beta}| \geq 3$.
\end{claim}

\begin{proof}
Towards a contradiction, let $\alpha \in \S$ be such that $|O_{\alpha + \beta}| = 2$ for all $\beta \in \S \setminus \bra{\alpha}$. Let $s = |\S_+|$ and $t = |\S_-|$.
We now prove that $s$ and $t$ must both be even.

Since $f$ is plateaued, Equation~\eqref{eqn: parseval} implies that $|\wh{f}(\gamma)| = 1/\sqrt{k}$ for all $\gamma \in \S$. By Observation~\ref{obs: gran} we know that $1/\sqrt{k} = c/2^n$ for some $c \in \Z$. This implies that $k = 2^{2n}/c^2$. Since $k$ is an integer, $c$ must be a power of 2, and hence $k = 2^{2h}$ for some $h > 1$ (since we assumed $k > 4$). 

Assume $f(1^n) = 1$ (else run the same argument with $f$ replaced by $-f$). This implies 
\[
\sum_{\gamma \in \S_+} \wh{f}(\gamma) - \sum_{\delta \in \S_-} \wh{f}(\delta) = 1.
\]
That is, $(s - t)/\sqrt{k} = 1$. Since $s + t = k$, this implies $s = \frac{k}{2} + \frac{\sqrt{k}}{2}$ and $t = \frac{k}{2} - \frac{\sqrt{k}}{2}$. Since $k = 2^{2h}$ for some $h > 1$ (since we assumed $k > 4$), $s$ and $t$ are both even. This is a contradiction in view of Claim~\ref{claim: trivialfoldodd}.
\end{proof}

We next use Claim~\ref{claim: 2foldingplat} to remove the assumption of $f$ being plateaued in the previous claim, which proves Theorem~\ref{thm: 3foldintro}.

\begin{proof}[Proof of Theorem~\ref{thm: 3foldintro}]
Towards a contradiction, suppose there exists $\alpha \in \S$ such that $|O_{\alpha + \beta}| = 2$ for all $\beta \in \S \setminus \bra{\alpha}$. Claim~\ref{claim: 2foldingplat} implies that $f$ must be plateaued. Next, Claim~\ref{claim: platfold} implies that there must exist $\gamma \in \S$ such that $|O_{\alpha + \gamma}| \geq 3$, which is a contradiction.
\end{proof}

We have shown that for any Boolean function $f : \f^n \to \pmone$, each $\alpha \in \S$ participates in at least one non-trivial folding direction. We next investigate if an element $\alpha \in \S$ can participate in exactly one non-trivial folding, and first prove Part~\ref{item: allfunctions} of Theorem~\ref{thm: singledirectionintro}. This states that if an element $\alpha \in \S$ only participates in one non-trivial folding direction, say $\delta$, then all elements of $\S$ must participate in the folding direction $\delta$.

\begin{proof}[Proof of Part~\ref{item: allfunctions} of Theorem~\ref{thm: singledirectionintro}]
Suppose there exist $(\alpha, \beta) \in \binom{\S}{2}$ such that $|O_{\alpha + \gamma}| = 2$ for all $\gamma \in \S \setminus \bra{\alpha, \beta}$. Let $|O_{\alpha + \beta}| = q > 2$ (Theorem~\ref{thm: 3foldintro} shows $q$ cannot equal 2), and say $O_{\alpha + \beta} = \bra{(\alpha, \beta), (\zeta_{1, 1}, \zeta_{1, 2}), (\zeta_{2, 1}, \zeta_{2, 2}), \dots, (\zeta_{{q-1}, 1}, \zeta_{{q-1}, 2})}$. Let $\P := \bra{\zeta \in \S ~\middle|~ \zeta~\textnormal{appears in a pair in}~O_{\alpha + \beta}}$. Clearly $|\P| = 2q$. Assume without loss of generality (justified by Proposition~\ref{prop:wlog}) that $\wh{f}(\alpha)>0$ and $\wh{f}(\beta)>0$.

If $\P = \S$, then $|O_{\alpha + \beta}| = k/2$, which proves the theorem. Otherwise pick any $\delta \in \S \setminus \P$. 
Since $|O_{\alpha + \delta}| = 2$ by assumption, let $O_{\alpha + \delta} = \bra{(\alpha, \delta), (\mu, \nu)}$, i.e.~$\alpha + \delta = \mu + \nu$. Next note that $\mu\neq\beta$, since otherwise $\alpha + \beta = \delta + \nu$, which implies that $(\delta, \nu) \in O_{\alpha + \beta}$, contradicting our choice of $\delta$.
Similarly, $\nu \neq \beta$. Thus, we also obtain
\begin{equation}\label{eqn: 2dirs}
O_{\alpha + \mu} = \bra{(\alpha, \mu), (\delta, \nu)}, \qquad O_{\alpha + \nu} = \bra{(\alpha, \nu), (\delta, \mu)}.
\end{equation}
Arrange $\alpha, \delta, \mu$ and $\nu$ in non-increasing order of their magnitudes; let the resultant ordering be $\delta_1, \delta_2, \delta_3, \delta_4$. Thus we have, $|\wh{f}(\delta_1)| \geq |\wh{f}(\delta_2)| \geq |\wh{f}(\delta_3)| \geq |\wh{f}(\delta_4)|$.
By Equation~\eqref{eqn: boolstruct} and Equation~\eqref{eqn: 2dirs}, we have $|\wh{f}(\delta_1)|| \wh{f}(\delta_2)| = |\wh{f}(\delta_3)||\wh{f}(\delta_4)|$, from which we conclude $|\wh{f}(\delta_1)| = |\wh{f}(\delta_2)| = |\wh{f}(\delta_3)| = |\wh{f}(\delta_4)|$.
Hence,
\begin{equation}\label{eqn: almostallsame}
|\wh{f}(\alpha)| = |\wh{f}(\delta)|~\text{for all}~\delta \in \S \setminus \P.
\end{equation}

Next, consider any $(\zeta_{i, 1}, \zeta_{i, 2}) \in O_{\alpha + \beta} \setminus \bra{(\alpha, \beta)}$. Since $|O_{\alpha + \zeta_{i, 1}}| = 2$ by our hypothesis, we have $O_{\alpha + \zeta_{i, 1}} = \bra{(\alpha, \zeta_{i, 1}), (\beta, \zeta_{i, 2})}$. Similarly, we have $O_{\alpha + \zeta_{i, 2}} = \bra{(\alpha, \zeta_{i, 2}), (\beta, \zeta_{i, 1})}$.
From Equation~\eqref{eqn: boolstruct} we have 
\begin{equation}\label{eqn: alphabetazeta}
\wh{f}(\alpha) \wh{f}(\zeta_{i, 1}) = - \wh{f}(\beta)\wh{f}(\zeta_{i, 2}), \qquad \wh{f}(\alpha) \wh{f}(\zeta_{i, 2}) = -\wh{f}(\beta)\wh{f}(\zeta_{i, 1}).
\end{equation}
Multiplying each side of these equalities, and since we have assumed $\wh{f}(\alpha) > 0$ and $\wh{f}(\beta) > 0$, we obtain
\begin{equation}\label{eqn: alphabetasame}
\wh{f}(\alpha) = \wh{f}(\beta) = a, ~\textnormal{say}
\end{equation}
and
\begin{equation}\label{eqn: ppairssame}
|\wh{f}(\zeta_{i, 1})| = |\wh{f}(\zeta_{i, 2})| = a_i,~\textnormal{say, for all}~i \in [q - 1]. \end{equation}
By substituting $\gamma = \alpha + \beta$ in Equation~\eqref{eqn: boolstruct}, we have
\begin{equation}\label{eqn: alpha+beta}
\wh{f}(\alpha)\wh{f}(\beta) + \sum_{i \in [q - 1]}\wh{f}(\zeta_{i, 1})\wh{f}(\zeta_{i, 2}) = 0.
\end{equation}
We also conclude from Equations~\eqref{eqn: alphabetazeta} and \eqref{eqn: alphabetasame} that $\wh{f}(\zeta_{i, 1}) = - \wh{f}(\zeta_{i, 2})$ for all $i \in [q]$. Hence, Equations~\eqref{eqn: alpha+beta},\eqref{eqn: alphabetasame}, and \eqref{eqn: ppairssame} imply
\begin{equation}\label{eqn: toomanyequations}
a^2 - \sum_{i \in [q]}a_i^2 = 0.
\end{equation}

We next use Parseval's identity to deduce the value of $a$. We have
\begin{align*}
1 & = \sum_{\delta \in \S}\wh{f}(\delta)^2 \tag*{by Equation~\eqref{eqn: parseval}}\\
& = \sum_{\delta \notin \P} \wh{f}(\delta)^2 + \sum_{\delta \in \P} \wh{f}(\delta)^2\\
& = (k - 2q)a^2 + \sum_{\delta \in \P} \wh{f}(\delta)^2 \tag*{since $|\P| = 2q$, and by Equations~\eqref{eqn: almostallsame} and \eqref{eqn: alphabetasame}}\\
& = (k - 2q + 2)a^2 + 2\sum_{i \in [q-1]}a_i^2 \tag*{by Equations~\eqref{eqn: alphabetasame} and \eqref{eqn: ppairssame}}\\
& = (k - 2q + 4)a^2 \tag*{by Equation~\eqref{eqn: toomanyequations}}
\end{align*}

Hence,
\begin{equation}\label{eqn: aval}
a = \frac{1}{\sqrt{k - 2q + 4}}.
\end{equation}

By Observation~\ref{obs: gran} we have $\frac{1}{\sqrt{k - 2q + 4}} = \frac{c}{2^n}$ for some $c \in \Z$. Thus, $k - 2q + 4 = 2^{2n}/c^2$. Since each side of this equality is an integer, $c$ must be a power of 2, and hence
\begin{equation}\label{eqn: kevenpower}
k = 2^{2h} + 2q - 4~\text{for some}~h \in \Z.
\end{equation}
Define $s :=|\S_+|$ and $t := |\S_-|$. Clearly $s + t = k$. We next show that $s - t = f(1^n)/a$. We have
\begin{align}
f(1^n)=\sum_{\delta \in \S_+} |\wh{f}(\delta)|-\sum_{\delta \in \S_-} |\wh{f}(\delta)|.
\label{eqn:f1n}
\end{align}
By our earlier discussion, for each $i \in [q-1]$ exactly one of $(\zeta_{i,1},\zeta_{i,2})$ is in $\S_+$ and the other is in $\S_-$. Furthermore, $|\wh{f}(\zeta_{i,1})|=|\wh{f}(\zeta_{i,2})|$. Thus, the contribution of parities in $\P \setminus\bra{\alpha, \beta}$ to the above sum is $0$. Hence it is enough to account for the contribution from $(\S \setminus \P) \cup \bra{\alpha, \beta}$, which contains $s - (q-1)$ elements of $\S_+$ and $t - (q-1)$ elements of $\S_-$. Now, recall from Equations~\eqref{eqn: almostallsame} and~\eqref{eqn: alphabetasame} that $\wh{f}(\alpha), \wh{f}(\beta)$ and the absolute values of coefficients of parities in $\S \setminus \P$ are all equal to $a$. Thus, we have from Equation~(\ref{eqn:f1n}),
\begin{align}
f(1^n) & =(s-(q-1))a-(t-(q-1))a \nonumber \\
\implies s-t & =\frac{f(1^n)}{a}. \nonumber
\end{align}
Thus we have that
\begin{align}
s&=\frac{s+t}{2}+\frac{s-t}{2} \nonumber \\
&=\frac{k}{2}+\frac{f(1^n)}{2a} \nonumber\\
& = \frac{k}{2} + \frac{f(1^n)\sqrt{k - 2q + 4}}{2} \tag*{by Equation~\eqref{eqn: aval}} \nonumber\\
& = 2^{2h-1} + q - 2 + f(1^n)\cdot 2^{h-1},\label{eqn: s}
\end{align}
where the last equality follows from Equation~\eqref{eqn: kevenpower}.
Now, if $h = 1$, then $|O_{\alpha + \beta}| = q = k/2$ (by Equation~\eqref{eqn: kevenpower}), which proves the theorem. The case $h=0$ is ruled out by Equation~\eqref{eqn: kevenpower} along with the fact that $q \leq k/2$. In the rest of the proof we assume that $h > 1$ and derive a contradiction. By Equation~(\ref{eqn: s}) the parity of $s$ equals the parity of $q$.

Next, by a different counting argument, we show that the parity of $s$ does not equal the parity of $q$, which will yield the desired contradiction. This counting mimics the argument in the proof of Claim~\ref{claim: trivialfoldodd}. We recall here that by our assumption $\alpha, \beta \in \S_+$.

Consider the set of unordered triples
\[
T = \bra{(\mu, \gamma, \delta) \in \binom{\S \setminus \bra{\alpha}}{3} ~\middle|~ \alpha + \mu + \gamma + \delta = \emptyset~\text{and}~\mu, \gamma, \delta \in \S}.
\]
As in the proof of Claim~\ref{claim: trivialfoldodd}, define
\[
T_+ := \bra{(\mu, \gamma, \delta) \in T ~\middle| ~\textnormal{at least one of}~ \wh{f}(\mu), \wh{f}(\gamma) , \wh{f}(\delta) ~\textnormal{is positive}}.
\]

First consider any parity $\nu \in \S \setminus \P$. Since $|O_{\alpha + \nu}| = 2$ by our definition of $\P$, say $O_{\alpha + \nu} = \bra{(\alpha, \nu), (\zeta, \xi)}$. Thus, $\nu$ appears in exactly one triple in $T_+$, namely $(\nu, \zeta, \xi)$. Since $\alpha \in \S_+$, Equation~\eqref{eqn: boolstruct} implies that each such triple contributes to either zero or two elements of $\S_+$. Hence these triples, along with the elements $\alpha$ and $\beta$, account for an even number of elements of $\S_+$. Now by our earlier discussion each pair $(\zeta_{i,1},\zeta_{i,2}) \in O_{\alpha+\beta} \setminus \bra{(\alpha, \beta)}$ for $i \in [q-1]$ contributes exactly one parity to $\S_+$. Hence, the parity of $s$ equals the parity of $(q-1)$, which is a contradiction.

\end{proof}

One might expect that a stronger statement than the one in Part~\ref{item: allfunctions} of Theorem~\ref{thm: singledirectionintro} is possible. For example, given any Boolean function $f : \f^n \to \pmone$ and $\alpha \in \S$, it is feasible that there must exist at least two parities $\beta \in \S$ with $|O_{\alpha + \beta}| > 2$.
We rule out this possibility in the next proof, hence showing that the hypothesis in Part~\ref{item: allfunctions} of Theorem~\ref{thm: singledirectionintro} can indeed by satisfied. We exhibit an explicit function that witnesses this.

\begin{proof}[Proof of Part~\ref{item: existsfunction} of Theorem~\ref{thm: singledirectionintro}]

Recall from Definition~\ref{def:addressing} that $\ADD_k$ denotes the addressing function on $n := ((\log k)/2 + \sqrt{k})$ many input bits. Let $(\mathbf{x}, \mathbf{y})$ be a generic input to $\ADD_k$, where $\mathbf{x} \in \f^{(\log k)/2}$ and $\mathbf{y} \in \f^{\sqrt{k}}$.

Define $g : \f^{n + 2} \to \pmone$ by
\begin{equation}\label{eqn: ifelsefn}
g(z_1, z_2, \mathbf{x}, \mathbf{y}) : = \begin{cases}
b_1 & \text{if } \ADD_k(\mathbf{x}, \mathbf{y}) = 1\\
b_2 & \text{if } \ADD_k(\mathbf{x}, \mathbf{y}) = -1,
\end{cases}
\end{equation}
where $b_1 = (-1)^{z_1}$ and $b_2 = (-1)^{z_2}$.
Now,
\begin{align*}
g(z_1, z_2, \mathbf{x}, \mathbf{y}) & = b_1 \left(\frac{1 + \ADD_k(\mathbf{x}, \mathbf{y})}{2}\right) + b_2 \left(\frac{1 - \ADD_k(\mathbf{x}, \mathbf{y})}{2}\right)\\
& = \frac{b_1}{2} + \frac{b_2}{2} + \frac{b_1 \cdot \ADD_k(\mathbf{x}, \mathbf{y})}{2} - \frac{b_2 \cdot \ADD_k(\mathbf{x}, \mathbf{y})}{2}.
\end{align*}
For the rest of this proof, we view elements of $\f^{n+2}$ as subsets of the input variables, via the natural correspondence. Under this equivalence, addition of elements in $\f^{n+2}$ corresponds to the symmetric difference of their respective sets.

Let $\S$ denote the Fourier support of $\ADD_k$. Let $\S'$ denote the Fourier support of $g$, and let $\alpha = \bra{z_1}$ and $\beta = \bra{z_2}$. Since $f$ is non-constant, we have
\[
\S' = \bra{\alpha, \beta} \cup \bra{\alpha \cup T ~\middle|~ T \in \S} \cup \bra{\beta \cup T ~\middle|~ T \in \S}.
\]

By Equation~\eqref{eqn: addexpansion}, every element in $\S$ contains exactly one variable from $\mathbf{y}$. Hence, for any $T_1, T_2 \in \S$, the set $T_1 \triangle T_2$ contains either no variables from $\mathbf{y}$, or exactly 2 variables from $\mathbf{y}$. Thus for any $T, T_1, T_2 \in \S$, we have $T \neq T_1 \triangle T_2$.
It now follows that for any $T \in \S$,
\begin{align*}
O_{\alpha \triangle (\alpha \cup T)} & = O_{T} =  \bra{(\alpha, \alpha \cup T), (\beta, \beta \cup T)},\\
O_{\alpha \triangle (\beta \cup T)} & = O_{\alpha \cup \beta \cup T} = \bra{(\alpha, \beta \cup T), (\beta, \alpha \cup T)}.
\end{align*}
Hence, $|O_{\alpha \triangle \gamma}| = 2$ for all $\gamma \in \S' \setminus \bra{\alpha, \beta}$, which proves the theorem.

\end{proof}

\section*{Acknowledgements}
We thank Prahladh Harsha, Srikanth Srinivasan, Sourav Chakraborty and Manaswi Paraashar for useful discussions.

\bibliography{bibo}

\appendix

\section{Ruling out sufficiency of Proposition~\ref{prop:tit}}\label{sec: app}

In this section, we prove that the conditions in Proposition~\ref{prop:tit} are not sufficient for a function to be Boolean.  To the best of our knowledge, ours is the first work to show this.
\begin{theorem}\label{thm: ceg}
There exists a set $\S \subseteq \f^n$ such that $|O_{\alpha + \beta}| \geq 2$ for all $(\alpha, \beta) \in \binom{\S}{2}$, but $\S$ is not the Fourier support of any Boolean function $f : \f^n \to \pmone$.
\end{theorem}
For sets $A, B \subseteq [n]$, let $A \triangle B$ denote the symmetric difference of the sets $A$ and $B$. For $x \in \reals \setminus \bra{0}$, define $\mathrm{sgn}(x) := -1$ if $x < 0$, and $\mathrm{sgn}(x) := 1$ if $x > 0$.
\begin{proof}
For the purpose of this proof, we
require the natural equivalence between elements of $\f^n$ and subsets of $[n]$. Under this equivalence, the sum of two elements in $\f^n$ corresponds to the symmetric difference of the corresponding sets in $[n]$.
The following is a property of symmetric difference. For any sets $A, B, C, D \subseteq [n]$,
\begin{equation}\label{eqn: symmdiffprop}
A \triangle B = C \triangle D \iff A \triangle C = B \triangle D. 
\end{equation}
Hence it suffices to exhibit a collection $\S$ of subsets of $[n]$ such that for all $(S, T) \in \binom{S}{2}$, there exist $(U, V) \neq (S, T) \in \binom{\S}{2}$ with $S \triangle T = U \triangle V$, and $\S$ is not the Fourier support of any Boolean function $f : \f^n \to \pmone$.
To this end, consider the set 
\[
\S = \bra{\bra{1}, \dots, \bra{n}, \bra{1, 2, n}, \dots, \bra{1, n-1, n}}.
\]
Below we list out all equivalence classes of $\binom{\S}{2}$.
For any distinct $i, j \in \bra{2, 3, \dots, n - 1}$ we have
\begin{align*}
\bra{i} \triangle \bra{j} & = \bra{1, i, n} \triangle \bra{1, j, n}.
\end{align*}
Thus
\begin{align}\label{eqn: eqclassceg}
O_{\bra{i} \triangle \bra{j}} = \bra{(\bra{i}, \bra{j}), (\bra{1, i, n}, \bra{1, j, n})} \quad \forall i, j \in \bra{2, 3, \dots, n - 1}.
\end{align}
For any $i \in \bra{2, 3, \dots, n-1}$ we have
\begin{align*}
\bra{1} \triangle \bra{i} & = \bra{n} \triangle \bra{1, i, n},\\
\bra{n} \triangle \bra{i} & = \bra{1} \triangle \bra{1, i, n}.
\end{align*}
We also have
\begin{align*}
\bra{1} \triangle \bra{n} = \bra{i} \triangle \bra{1, i, n} \quad\text{for all }i \in \bra{2, 3, \dots, n-1}.
\end{align*}
Along with Equation~\eqref{eqn: symmdiffprop}, these establish the fact that $|O_{\alpha + \beta}| \geq 2$ for all $(\alpha, \beta) \in \binom{\S}{2}$. We now provide a proof of the fact that $\S$ cannot be the Fourier support of any Boolean function. Consider the following six sets.
\begin{align*}
S_1 = \bra{2}, S_2 = \bra{3}, S_3 = \bra{4}, S_4 = \bra{1, 2, n}, S_5 = \bra{1, 3, n}, S_6 = \bra{1, 4, n}.
\end{align*}

If $\S$ is the support of a Boolean function, then Equation~\eqref{eqn: boolstruct} holds true. Equation~\eqref{eqn: eqclassceg} then implies
\begin{align*}
\wh{f}(S_1)\wh{f}(S_2) + \wh{f}(S_4)\wh{f}(S_5) & = 0,\\
\wh{f}(S_1)\wh{f}(S_3) + \wh{f}(S_4)\wh{f}(S_6) & = 0,\\
\wh{f}(S_2)\wh{f}(S_3) + \wh{f}(S_5)\wh{f}(S_6) & = 0.
\end{align*}
Let $s_i = \mathrm{sgn}(\wh{f}(S_i))$ for $i \in [6]$. Thus,
\begin{align*}
s_1s_2 & = -s_4s_5\\
s_1s_3 & = -s_4s_6\\
s_2s_3 & = -s_5s_6.
\end{align*}
Multiplying out the left hand sides and right hand sides of the above, we obtain $1 = -1$, which is a contradiction. Hence $\S$ cannot be the support of any Boolean function.
\end{proof}

\section{Folding properties of the addressing function}\label{sec: addfold}
In this section analyze folding properties of the addressing function. 

\begin{claim}
The function $\ADD_k :\f^{\frac{1}{2}\log k + \sqrt{k}} \to \pmone$ as defined in Definition~\ref{def:addressing} is
\begin{enumerate}
    \item $(1, 1/2 - o(1))$-folding, and
    \item not $(\Omega(1), \ell)$-folding for any $\ell \geq 1/2$.
\end{enumerate}
\end{claim}
For this proof, we view elements of $\f^{\frac{1}{2}\log k + \sqrt{k}}$ as subsets of the set of variables, via the natural equivalence. Addition over $\f^n$ corresponds to symmetric difference of the respective sets.
\begin{proof}

For $a \in [\sqrt{k}]$, define the function $\1_{a} : \f^{\frac{1}{2}\log k} \to \zone$ by \begin{align*}
\1_{a}(x) = 
\begin{cases}
1 & \int(x) = a\\
0 & \text{otherwise}.
\end{cases}
\end{align*}

The Fourier expansion of $\ADD_k$ is given by
\begin{align}\label{eqn: addexpansion}
\sum_{a \in [\sqrt{k}]}\1_a(x) (-1)^{y_{a}}.
\end{align}
It can be verified that for any $a \in [\sqrt{k}]$, the Fourier support of $\1_a$ consists of all subsets of $X:=\bra{x_1, \dots, x_{\frac{1}{2} \log k}}$.
In light of this, the Fourier support $\S$ of $\ADD_k$ is given by
\begin{equation}\label{eqn: addsupport}
\S = \bigcup_{a \in [\sqrt{k}], {M \subseteq X}} \bra{M \cup \bra{y_a}}.
\end{equation}
It is not hard to verify that $|\S| = k$.
Let $\alpha, \beta$ be two distinct elements of $\S$. We now determine $O_{\alpha + \beta}$.
\begin{description}
\item[Case 1]: \emph{$\alpha = M_1 \cup \bra{y_a}$ and $\beta = M_2 \cup \bra{y_a}$ for distinct $M_1, M_2 \subseteq X$, and any {$a \in [\sqrt{k}]$}}.\\
In this case, 
\[
O_{\alpha \triangle \beta} = \bra{(M'_1 \cup \bra{y_b}, M'_2 \cup \bra{y_b})~\middle|~M'_1 \triangle M'_2 = M_1 \triangle M_2, b \in [\sqrt{k}]}.
\]
Hence, $|O_{\alpha \triangle \beta}| = k$.

\item[Case 2]: \emph{$\alpha = M_1 \cup \bra{y_a}$ and $\beta = M_2 \cup \bra{y_b}$ for distinct {$a, b \in [\sqrt{k}]$}, and any $M_1, M_2 \subseteq X$.}\\
In this case,
\[
O_{\alpha \triangle \beta} = \bra{(M'_1 \cup \bra{y_a}, M'_2 \cup \bra{y_b})~\middle|~M'_1 \triangle M'_2 = M_1 \triangle M_2}.
\]
Hence, $|O_{\alpha \triangle \beta}| = \sqrt{k}$.
\end{description}
From both cases above we conclude that for all $(\alpha, \beta) \in \binom{\S}{2}$, we have $|O_{\alpha \triangle \beta}| \geq \sqrt{k} = k^{\ell} + 1$ for $\ell = 1/2 - o(1)$. This immediately proves the first part of the claim.

It is easy to verify that the number of pairs $(\alpha, \beta) \in \binom{\S}{2}$ that fall under the second case above is $(1 - o(1))\binom{k}{2}$. For all these pairs, we have $|O_{\alpha \triangle \beta}| = \sqrt{k} < k^{1/2} + 1$. Hence the number of pairs $(\alpha, \beta) \in \binom{\S}{2}$ with $|O_{\alpha \triangle \beta}| \geq k^{1/2} + 1$ is $o(1) \cdot \binom{k}{2}$. This proves the second part of the claim.

\end{proof}

\end{document}